\documentclass[10pt,a4paper]{article}

\DeclareMathAlphabet{\mathcal}{OMS}{cmsy}{m}{n}
\newcommand\bcmdtab{\noindent\bgroup\tabcolsep=0pt%
  \begin{tabular}{@{}p{10pc}@{}p{20pc}@{}}}
\newcommand\ecmdtab{\end{tabular}\egroup}

\usepackage{bm}
\usepackage{mathptmx}
\usepackage{multirow}
\usepackage{xspace}
\usepackage{amsmath}
\usepackage{graphicx}
\usepackage{amsthm}
\usepackage{amssymb}
\usepackage[usenames, dvipsnames]{color}
\usepackage{xspace}
\usepackage{comment}
\usepackage{url}
\usepackage{bigstrut}
\usepackage{tikz}

\newcommand{\rar}{\rightarrow}
\newcommand{\sat}{\mathit{SAT}}

\newcommand{\nop}[1]{}

\newif\ifdraft\drafttrue
\newif\ifinlineref\inlinereffalse
\newif\iffinal\finalfalse
\newif\ifextended\extendedfalse
\newif\ifdotikz\dotikzfalse
\newif\ifmakeallproofsinline\makeallproofsinlinefalse
\newif\ifeditors\editorstrue

\inlinereftrue
\editorsfalse
\finaltrue

\extendedtrue
\newif\iftechreport\techreportfalse

\usepackage{times}
\usepackage[scaled]{helvet}
\usepackage{courier}
\usepackage{xcolor}
\usepackage{multirow}

\usepackage{latexsym}

\usepackage{url}
\usepackage{paralist}
\usepackage{enumerate}
\usepackage{xspace}
\usepackage{bm}
\usepackage{amsmath}
\usepackage{amssymb}
\usepackage{mathptmx}

\usepackage{ifpdf}
\ifpdf
\usepackage{microtype}
\fi

\ifdraft
\usepackage{pdfsync}
\usepackage{srcltx}
\else
\fi

\newif\ifdotikz\dotikzfalse

\dotikztrue

\usepackage{graphicx}

\ifdotikz
\usepackage{tikz}
\pgfrealjobname{main}
\usetikzlibrary{calc}
\else
\long\def\beginpgfgraphicnamed#1#2\endpgfgraphicnamed{\includegraphics{#1}}
\fi

\usepackage{verbatim}
\usepackage{xspace}
\usepackage[ruled,linesnumbered,vlined]{algorithm2e}
\usepackage{epsfig}
\usepackage{epstopdf}
\usepackage{xcolor}
\usepackage{caption}
\usepackage{subcaption}
\usepackage[normalem]{ulem}

\newtheorem{thm}{Theorem}[section]
\newtheorem{corollary}{Corollary}[section]
\newtheorem{lemma}{Lemma}[section]
\newtheorem{example}{Example}[section]

\newcommand{\ignore}[1]{}
\newcommand{\kcnfs}{\textsc{kcnfs}\xspace}
\newcommand{\glucose}{\textsc{glucose}\xspace}
\newcommand{\lingeling}{\textsc{lingeling}\xspace}
\newcommand{\bqcegar}{\textsc{bq-Cegar}\xspace}
\newcommand{\aigsolve}{\textsc{Aigsolve}\xspace}
\newcommand{\wasp}{\textsc{wasp}\xspace}
\newcommand{\clasp}{\textsc{clasp}\xspace}
\newcommand{\rareqs}{\textsc{RAReQS}\xspace}
\newcommand{\aqua}{\textsc{AQuA-S2V}\xspace}

\newcommand{\cQ}{\mathcal{Q}}

\newcommand{\cF}{\mathcal{F}}

\newcommand{\gctd}{\mathit{gctd}}
\newcommand{\sgctd}{\mathit{sgctd}}
\newcommand{\ctd}{\mathit{ctd}}
\newcommand{\dlp}{\mathit{dlp}}
\newcommand{\qbf}{\mathit{Q}}
\newcommand{\cnf}{\mathit{C}}
\newcommand{\dnf}{\mathit{D}}

\def\<{\langle}
\def\>{\rangle}
\def\int{\mathit{Int}}

\newcommand{\rra}{\rightarrow}

\newcommand{\lla}{\leftarrow}

\newcommand{\n}{\mathit{not}\;}

\title{New Models for Generating Hard\\ Random Boolean Formulas and\\ Disjunctive Logic Programs\footnote{Some of the results were presented in preliminary form at IJCAI~2017~\cite{IJCAIVERSION}.}}
\date{}
\author{Giovanni Amendola$^1$, Francesco Ricca$^1$, Miroslaw Truszczynski$^2$ \\
         $^1$University of Calabria, Italy,  
         $\{$amendola,ricca$\}$@mat.unical.it\\
         $^2$University of Kentucky, USA, 
         mirek@cs.uky.edu}

\usepackage{color}
\usepackage{todonotes}

\begin{document}

\label{firstpage}

\maketitle

\begin{abstract}
We propose two models of random quantified boolean formulas and their natural
random disjunctive logic program counterparts. 
The models extend the standard models of random $k$-CNF formulas 
and the Chen-Interian model of random 2QBFs.
The first model controls the generation of programs and QSAT formulas 
by imposing a specific structure on rules and clauses, respectively. 
The second model is based on a family of QSAT formulas in a non-clausal form. 
We provide theoretical bounds for the phase transition region in our models, 
and show experimentally the presence of the easy-hard-easy pattern
and its alignment with the location of the phase transition. 
We show that boolean formulas and logic programs from our models are 
significantly harder than those obtained from the standard $k$-CNF and 
Chen-Interian models, and that their combination yields formulas and 
programs that are ``super-hard'' to evaluate. We also provide evidence
suggesting that formulas from one of our models are well suited for 
assessing solvers tuned to real-world instances. Finally, it is noteworthy
that, to the best of our knowledge, our models and results on random 
disjunctive logic programs are the first of their kind. 
\end{abstract}


\section{Introduction}
Models for generating random instances of search problems have 
received much attention from the artificial intelligence community 
in the last twenty years. The results obtained for boolean 
satisfiability (SAT)~\cite{Achlioptas09,SelmanML96} and constraint satisfaction 
(CP)~\cite{DBLP:conf/cp/Mitchell02} have had a major impact on the development of fast and robust
solvers, significantly expanding their range of effectiveness as general 
purpose tools for solving hard search and optimization problems arising in
AI, and scientific and engineering applications. They also revealed an
intriguing phase-transition phenomenon often associated with the inherent 
hardness of instances, and provided theoretical and experimental basis for 
a good understanding of the ``region'' where the phase-transition occurs.

Models of random propositional formulas and QBFs that can reliably generate 
large numbers of instances of a \emph{desired} hardness are 
important~\cite{GentW99}. Inherently hard instances for SAT and QBF solvers 
are essential for designing and testing search methods employed by 
solvers~\cite{Achlioptas09},
and are used to assess their performance in solver competitions~%
\cite{JarvisaloBRS12,NarizzanoPT06,CalimeriGMR16}.
On the flip side, large collections of \emph{easy} instances 
support the so-called \emph{fuzz} testing, used to reveal problems in solver
implementation, as well as defects in solver design~\cite{BrummayerLB2010}.

Previous work on models of random formulas  focused on 
random CNF formulas and random prenex-form QBFs with the matrix in CNF or DNF 
(depending on the quantifier sequence). 
The fixed-length
clause model of $k$-CNF formulas and its 2QBF extension
have been especially well studied. 
Formulas in the fixed-length clause model consist of $m$ clauses over 
a (fixed) set of $n$ variables, each clause 
with 
$k$ non-complementary literals. All formulas are assumed to be equally 
likely. For that model it is known that there are reals $\rho_l(k)$ and
$\rho_u(k)$ 
such that if $m/n < \rho_l(k)$, a formula from the model is almost surely 
satisfiable (SAT), and if $m/n > \rho_u(k)$, almost surely unsatisfiable
(UNSAT).\footnote{We give a precise statement of these properties
in Section \ref{s21}.} It is conjectured that $\rho_l(k) = \rho_u(k)$. That conjecture 
is still open. However, it holds asymptotically, i.e., the two bounds 
converge to each other with $k\rightarrow \infty$ \cite{Achlioptas2}. For 
the best studied case of $k=3$, we have $\rho_l(3)\geq 3.52$ \cite{KaporisKL03}
and $\rho_u(3)\leq 4.49$ \cite{DiazKMP09}, and 
experiments show that the phase transition ratio $m/n$ is
close to $4.26$ \cite{CrawfordA96}. 
Important for the solver design and testing is that instances 
from the phase transition region are hard 
and those from regions on 
both sides of the phase transition are easy, a property called the 
\emph{easy-hard-easy} pattern \cite{MitchellSL92} or, more accurately,
the ``easy-hard-less hard'' pattern \cite{CoarfaDASV00}.   
Empirical studies suggest that SAT solvers devised for solving random formulas 
are usually not effective with real world instances; \emph{vice versa}
solvers for industrial instances are less efficient on random 
formulas~\cite{JarvisaloBRS12}. This is often attributed 
to some form of (hidden) structure present in industrial problems that 
solvers designed for industrial applications can 
exploit~\cite{DBLP:conf/aaai/AnsoteguiBLM08}. Finding models to generate
random formulas with ``structure'' that behave similarly to those arising 
in practice is an important challenge~\cite{DBLP:conf/cp/KautzS03}. 
Ansotegui et al. \cite{DBLP:conf/ijcai/AnsoteguiBL09} 
presented the first 
model that may have this property: despite the ``randomness'' of its
instances, they are better solved by solvers tuned to 
industrial applications.
More recently, Gir{\'{a}}ldez{-}Cru and Levy \cite{DBLP:journals/ai/Giraldez-CruL16} proposed a model of random
SAT based on the notion of \textit{modularity}, and showed that formulas
with high modularity behave similarly to industrial ones.

The fixed-length clause model was extended to QBFs by Chen and Interian~\cite{ChenI05}. 
In addition to $n$ and $m$ (understood as above), their model includes 
parameters controlling the structure of formulas. 
Once these parameters are fixed, similar properties as in 
the case of the $k$-CNF model emerge. There is a phase transition region 
associated with a specific value of the ratio $m/n$ (that does not 
depend on $n$) and the easy-hard-easy pattern can be experimentally
verified. 

These two models are based on formulas in normal forms. However, many 
applications give rise to formulas in non-normal forms motivating
studies of solvers of non-normal form formulas and QBFs, and raising
the need of models of random non-normal form formulas. The \emph{fixed-shape} 
model proposed by Navarro and Voronkov \cite{NavarroV05}, and studied 
by Creignou et al. \cite{CreignouES12}, is a response to that challenge. 
The model is similar to that of the $k$-CNF one (or its extensions to QBFs), 
but \emph{fixed shape} (and size) non-normal form formulas are used in place 
of $k$-clauses as the key building blocks. 
Experimental studies again show the phase-transition and the easy-hard-easy pattern.

Motivated by the work on random SAT and QBF models, 
researchers proposed models of random logic programs, 
and obtained empirical and theoretical
results concerning their properties 
\cite{ZhaoL03,wst05,NamasivayamT09,WangWM15,DBLP:journals/tocl/WenWSL16}. 
Those results are limited to \textit{non-disjunctive} logic programs. No 
models for \textit{disjunctive} logic programs have been proposed so far. 
Such results would be of substantial interest to answer set programming 
(ASP) \cite{BrewkaET11}, a popular computational formalism based on 
disjunctive logic programs.

In this paper we propose two models of random QBF formulas and 
the corresponding models of disjunctive logic programs. 
First, we propose a \emph{controlled} version of the Chen-Interian model 
in which CNF formulas that are used as matrices are subject to additional 
conditions restricting their structure. 
Second, we propose \emph{multi-component} versions of the earlier models.%
\footnote{When we use the term ``multi-component model,'' we        
understand that the notion is parameterized by the underlying ``standard,''
or single-component, model.}
In the multi-component models, propositional formulas and matrices of QBFs 
are disjunctions of $t$ $k$-CNF formulas (either standard or ``controlled'').
They are not formulas from the fixed-shape model of Navarro and
Voronkov, as their building blocks (CNF or DNF formulas) do not have a 
fixed size. In each case, the standard translation from QBFs to disjunctive
programs suggests random models for the latter.

For the new models, we present theoretical bounds on the region where 
the phase transition is located, and study experimentally their behavior.
In our experiments, we consider several ASP, SAT and QBF solvers
to exclude any possible bias that could be an artifact of a particular solver.
We study the regions of hardness for the models and show empirically that 
they lie within their phase transition regions. 
We compare the hardness of the controlled model with the corresponding Chen-Interian model
and find that the former can generate formulas that are significantly harder.
For the multi-component versions of the standard random CNF and the Chen-Interian models 
we study hardness as a function of the ratio $m/n$ and of the number of components $t$.
The latter was of main interest to us. The results show that 
the multi-component model allows for controlling hardness of formulas
and programs in such a way that, even when the number of variables is 
fixed, raising $t$ may result in instances that are orders of magnitude harder to evaluate.  
Moreover, we show that the combination of controlled and multi-component model allows 
to generate instances that are ``super-hard'' to evaluate.

As Ans{\'{o}}tegui et al. \cite{DBLP:conf/ijcai/AnsoteguiBL09}, 
we compare SAT/QBF solvers designed for 
random instances with 
those designed for real-world ones. We find that for $t \geq 2$ our 
models generate instances better solved by solvers for real-world instances, 
and that the difference becomes more pronounced as $t$ grows. 
For disjunctive logic programs, 
we measure the effect of $t$ on processing them 
and show that $t$ allows us to control the amount of computation dedicated 
to stable model checking~\cite{DBLP:journals/ai/KochLP03}.

Our results provide new ways to generate hard and easy instances of 
propositional formulas, QBFs and disjunctive programs. Our models can
generate instances of increasing hardness with properties affecting 
solver performance in a similar way real-world instances do.
The results are particularly important to the development of disjunctive
ASP solvers, as no models for generating random disjunctive programs of 
desired hardness have been known before.

\section{Preliminaries}

A \emph{clause} is a set of literals that contains no pair of 
complementary literals. By a \emph{CNF formula} we mean an (ordered)
tuple of clauses with repetitions of clauses allowed. Disjunctions 
of CNF formulas are also assumed to be (ordered) tuples and they also 
allow repetitions. The dual concepts (such as DNF formulas) 
are defined similarly. 
In other models, CNF formulas are viewed as \emph{sets} of clauses,
and disjunctions of CNF formulas are viewed as sets of CNF 
formulas. However, assuming some reasonable limit on the number of 
clauses in a formula, and assuming in each case the uniform distribution, 
the two probabilistic models are asymptotically equivalent for properties that 
do not depend on the order (such as satisfiability). 
Specifically, as the number of atoms tends to infinity, the probability 
that such a property holds in one model and the corresponding probability 
for the other model converge to each other. (We offer
a technical justification for this claim in Appendix B.) 
Thus, there is no essential difference between the two models 
and we use them interchangeably.

By $\cnf(k,n,m)$ we denote the set of all $k$-CNF formulas consisting of
$m$ clauses over (some fixed) set of $n$ propositional variables. 
Similarly, $\dnf(k,n,m)$ stands for the set of all $k$-DNF formulas of
$m$ \emph{products} (conjunctions of non-complementary literals) over an 
$n$-element set of atoms. 

\subsection{The fixed-length clause model}\label{s21}
The model is given by
the set $\cnf(k,n,m)$ of CNF formulas, with all formulas assumed equally
likely. Formulas from the model can be generated 
by selecting $m$ $k$-literal clauses over a set of $n$ 
variables uniformly, independently and with replacement. As we noted,
the model is well understood. In particular, let us denote by 
$p(k,n,m)$ the probability that a random formula in $\cnf(k,n,m)$ is SAT. 
We define $\rho_l(k)$ to be the supremum over all
real numbers $\rho$ such that 
$\lim_{n\rra\infty} p(k,n,\lfloor \rho n\rfloor) =1$.
Similarly, we define $\rho_u(k)$ to be the infimum over all real numbers
$\rho$ such that 
$\lim_{n\rra\infty} p(k,n,\lfloor \rho n\rfloor) =0$.
%
As we mentioned, $\rho_l(k)$ and $\rho_u(k)$ are well defined. Moreover,
$\rho_l(k) \leq \rho_u(k)$ and, it is conjectured that $\rho_l(k)=\rho_u(k)$.
Experimental results agree with these
theoretical predictions.  

\subsection{The Chen-Interian model}
The model generates QBFs of the form $\forall X\exists Y F$. Sets $X$ and 
$Y$ are disjoint and contain all propositional variables that may appear 
in $F$. The sizes of $X$ and $Y$ are prescribed to some specific integers
$A$ and $E$, respectively. Moreover, each clause in $F$ contains $a$ 
literals over $X$ and $e$ literals over $Y$ for some specific values 
$a$ and $e$. We denote the set of all such CNF formulas $F$ with $m$
clauses by $\cnf(a,e;A,E;m)$. Clearly, $\cnf(a,e;A,E;m) \subseteq
\cnf(a+e,A+E,m)$. We write $Q(a,e;A,E;m)$ for the set of all QBFs 
$\forall X\exists Y F$, where $F\in \cnf(a,e;A,E;m)$. The Chen-Interian
model generates QBFs from $Q(a,e;A,E;m)$, with all formulas equally likely. 

Chen and Interian \cite{ChenI05} presented a comprehensive 
experimental study of the model. 
Let us denote by
$q(a,e;A,E;m)$ the probability that a random QBF from $\qbf(a,e;A,E;m)$
is true. Let $r > 0$ be fixed real. We set $\nu_l(a,e;r)$
to be the supremum over all real numbers $\nu$ such that 
$\lim_{n\rra\infty} q(a,e;A,E;\lfloor \nu n\rfloor) =1$,
where $A=\lfloor rE\rfloor$ and $n=A+E$. Similarly, we set $\nu_u(a,e;r)$
to be the infimum over all real numbers $\nu$ such that 
$\lim_{n\rra\infty} q(a,e;A,E;\lfloor \nu n\rfloor) =0$,
again with $A=\lfloor rE\rfloor$ and $n=A+E$. Chen and Interian proved the following result.

\begin{thm}
\label{ci1}
$\nu_l(a,e;r)$ and $\nu_u(a,e;r)$ are well defined.
\end{thm}

Clearly, $\nu_l(a,e;r)\leq \nu_u(a,e;r)$. Whether $\nu_l(a,e;r)=
\nu_u(a,e;r)$ is an open problem. The quantities $\nu_l(a,e;r)$ and 
$\nu_u(a,e;r)$ delineate the phase-transition region.
For QBFs generated from the model $Q(a,e;\lfloor rE\rfloor,E;\lfloor \nu 
n\rfloor)$ (with fixed $n$ and $r$), Chen and Interian experimentally 
observed the easy-hard-easy pattern as $\nu$ grows. They showed that the 
hard region is aligned 
with the phase transition, and that the same behavior emerges no matter what 
concrete $r$ is fixed as the ratio $A/E$.


\section{New models of random formulas and QBFs}
\label{sec:multicomponent}

We propose several variations of the models described above. They are based 
on two ideas. First, we impose an additional structure on clauses in CNF 
formulas
that serve as matrices of QBFs. Second, we consider disjunctions of CNF 
formulas both in the SAT and QBF setting.

\subsection{The controlled model}\label{sec:controlled}
To describe the model, we define first a version of a model of a random 
CNF formula. In this model, clauses are built 
of variables in a set $X\cup Y$, where $X \cap Y = \emptyset$; 
we set $|X|=A$ and $|Y|=E$. A formula in the model consists of $2A$ $k$-literal
clauses.
Each clause consists of a single literal over $X$ and $k-1$ literals over $Y$,
and for each literal over $X$ there is a single clause in the formula that 
contains it. A formula in this model is generated taking $2A$ $(k-1)$-literal
clauses
over $Y$ and extending each of them by a literal over $X$ (following some fixed
one-to-one mapping between the clauses and the literals over $X$). We denote 
this model (and the corresponding set of
formulas) by $C^{\ctd}(k,A,E)$. We write $Q^{\ctd}(k,A,E)$ for the model 
(and the set) of QBFs whose matrix is a formula from $C^{\ctd}(k,A,E)$. We 
refer to both models as \emph{controlled}. In our work we are primarily 
interested in the controlled model for QBFs. 

Clearly, $Q^\ctd(k,A,E) \subseteq Q(1,k-1;A,E;2A)$. Thus, the controlled 
model is related to the Chen-Interian model. The main difference is that
the clauses, while random with respect to existential variables are not
random with respect to universal variables. For each $x\in X$ there is
exactly one clause involving $x$ and exactly one clause involving $\neg x$.
Consequently, the number of clauses is $2A$ and, moreover, 
for every truth assignment to $X$, once we simplify the matrix accordingly, 
we are left with \emph{exactly} (hence, the term ``controlled'') $A$ 
$(k-1)$-literal clauses 
over $E$ variables. In contrast, in the case of the 
Chen-Interian model $Q(1,k-1;A,E;2A)$, similar simplifications leave us 
with $(k-1)$-CNF formulas with \emph{varying} number of clauses, with the 
\emph{average} number being $A$.
 
Let $q^\ctd(k,A,E)$ denote the probability that a random formula in 
$Q^{\ctd}(k,A,E)$ is true. As before, we define $\mu^\ctd_l(k)$
to be the supremum over all positive real numbers $\rho$ such that
$\lim_{E\rra\infty} q^\ctd(k,$ $\lfloor \rho E\rfloor,E) =1$,
and $\mu^\ctd_u(k)$ to be the infimum over all positive real numbers $\rho$ such that 
$\lim_{E\rra\infty} q^\ctd(k,$ $\lfloor \rho E\rfloor, E) =0$.

We will now derive bounds on $\mu^\ctd_l(k)$ and $\mu^\ctd_u(k)$ by 
exploiting results on random $(k-1)$-CNF formulas.
\begin{thm}\label{th:controlledThreshold} 
For every $k\geq 2$, $\mu^\ctd_l(k)\geq \frac{\rho_l(k-1)}{2}$ and $\mu^\ctd_u(k)\leq \rho_u(k-1)$.
\end{thm}
\begin{proof}

Let $\Phi\in Q^{\ctd}(k,A,E)$, 
$X=\{x_1,...,x_A\}$, and $Y=\{y_1,...,y_E\}$.
By the definition, $\Phi=\forall X \exists Y F$, where $F=C_1 \wedge\ldots\wedge C_{2A}$ 
is a $k$-CNF formula of $2A$ 
clauses $C_i= l_{i1}\vee\ldots\vee l_{ik}$, and where $l_{i1}$ is a literal 
over $X$ and $l_{i2},\ldots,l_{ik}$ are literals over 
$Y$. We define $C_i^Y=l_{i2}\vee\ldots\vee l_{ik}$ and $F^Y = C_1^Y\land\ldots\land C_{2A}^Y$. Moreover, for every
interpretation $I$ of $X$ we define $F|_I=\bigwedge\lbrace C_i^Y \mid C_i\in F \mbox{ and } I\not\models l_{i1} \}$.

Let us assume that $\Phi$ is selected from $Q^{\ctd}(k,A,E)$ uniformly at random. By the definition of the model $Q^\ctd(k,A,E)$, $F^Y$ can be regarded as selected from $C(k-1,2A,E)$ uniformly at random and, for each interpretation $I$ of $X$, $F|_I$ can be regarded as selected uniformly at random from $C(k-1,A,E)$. 

To derive an upper bound on $\mu^\ctd_u(k)$, let us fix an interpretation $I$ of $X$. Clearly, if $F|_I$ is unsatisfiable, then $\Phi$ is false. 
Let us choose any real $\rho > \rho_u(k-1)$. If $A/E\geq \rho$, the probability that $F|_I$ is unsatisfiable 
tends to 1 with $E$ and, consequently, the probability that $\Phi$ is false tends 
to 0 with $E$, too. It follows that if $\rho > \rho_u(k-1)$ and $A/E\geq \rho$, the probability 
that $\Phi$ is true tends to 0 with $E$.
Since $\rho$ is an arbitrary real such that $\rho > \rho_u(k-1)$, 
$\mu^\ctd_u(k)\leq \rho_u(k-1)$ follows.

To prove the lower bound, we observe that if the formula $F^Y$
is satisfiable, then for every interpretation $I$ of $X$, the formula 
$F|_I$ is satisfiable or, equivalently, $\Phi$ is true.
Let $\rho$ be a positive real number such that $\rho < \frac{\rho_l(k-1)}{2}$. 
By the definition of $\rho_l(k-1)$, if we assume that
$A/E \leq \rho$, that is, $2A/E\leq 2\rho < \rho_l(k-1)$, the probability that $F^Y$
is satisfiable tends to 1 with $E$. Thus, the probability that $\Phi$ is true tends 
to 1 with $E$. It follows that $\mu^\ctd_l(k)\geq \frac{\rho_l(k-1)}{2}$.
\end{proof}

It follows that as $\rho$ grows, the properties of $Q^\ctd(k,\lfloor \rho E\rfloor,
E)$ change. For small values of $\rho$, randomly selected QBFs are almost 
surely true. As $\rho$ grows beyond $\mu^\ctd_l(k)$ the proportion of false 
formulas grows until, eventually, when $\rho$ grows beyond $\mu^\ctd_u(k)$,
the formulas in the model are almost surely false.
Clearly, $\mu^\ctd_l(k) \leq \mu^\ctd_u(k)$. 
As in the other cases, the question whether 
$\mu^\ctd_l(k)=\mu^\ctd_u(k)$ is open.

\subsection{The multi-component models}
Let $\cF$ be a class of propositional formulas (or a model of a random formula).
By $t$-$\cF$ we denote the class of all disjunctions of $t$ formulas from 
$\cF$ (or a model generating disjunctions of random formulas from $\cF$). 
Similarly, if $\cQ$ is a class (model) of QBFs of the form $\forall X\exists Y 
F$, where $F \in \cF$, we write $t$-$\cQ$ for the class (model) of all QBFs 
of the form $\forall X\exists Y F$, where $F \in \mbox{$t$-$\cF$}$. We refer 
to models $t$-$\cF$ and $t$-$\cQ$ as \emph{multi-component}. For QBFs we also 
consider the dual model to $t$-$\cQ$, based on conjunctions of $t$ DNF 
formulas. It gives rise to a multi-component model of disjunctive logic 
programs via the Eiter-Gottlob translation. In all cases, we assume that 
formulas (QBFs, respectively) are equally likely.
 
We first observe that the multi-component model $t$-$\cnf(k,n,m)$ has similar 
satisfiability properties as $\cnf(k,n,m)$, and that the phase transition regions in the two 
models are closely related. Let $p_t(k,n,m)$ be the probability that a 
random formula in $t$-$\cnf(k,n,m)$ is SAT. Clearly, $p_1(k,n,m)=p(k,n,m)$.

\begin{thm}
\label{thm:4}
Let $t\geq 1$ be a fixed integer. Then, for every $\rho < \rho_l(k)$,
$\lim_{n\rra\infty} p_t(k,n,\lfloor \rho n\rfloor) =1$,
and for every $\rho>\rho_u(k)$,
$\lim_{n\rra\infty} p_t(k,n,\lfloor \rho n\rfloor) =0$.
\end{thm} 
\begin{proof}
As we discussed earlier, we can assume that our model actually generates 
ordered $t$-tuples of $\cnf(k,n,m)$ formulas (they represent disjunctions 
of $t$ formulas from the model $\cnf(k,n,m)$, where repetitions of disjuncts 
are allowed, and disjunctions differing in the order of disjuncts are viewed
as different). Thus, it is clear that 
\begin{equation}\label{eq:prob}
p_t(k,n,m) = 1 - (1 - p(k,n,m))^t.
\end{equation}
It follows that for every fixed $t$, and every $\rho$,
\[
\lim_{n\rra\infty} p_t(k,n,\lfloor \rho n\rfloor) =0 \quad\mbox{if and only if}\quad \lim_{n\rra\infty} p(k,n,\lfloor \rho n\rfloor) = 0
\]
and 
\[
\lim_{n\rra\infty} p_t(k,n,\lfloor \rho n\rfloor) =1 \quad\mbox{if and only if}\quad \lim_{n\rra\infty} p(k,n,\lfloor \rho n\rfloor) = 1.
\]
Thus, the assertion follows.
\end{proof}

Theorem \ref{thm:4} implies that if the phase transition conjecture holds for the single
component model $\cnf(k,n,m)$, it also holds for the multi-component model
$t$-$\cnf(k,n,m)$, and the threshold value is the same for every $t$. 

Theorem \ref{thm:4} describes the situation when $t$ is fixed and $n$
is large. When $n$ is fixed and $t$ grows, the identity (\ref{eq:prob})
shows that the region of the transition from SAT to UNSAT shifts to the 
right. (Of course, by Theorem \ref{thm:4}, once we stop growing $t$
and start increasing $n$ again, the phase transition region will move back
to the left.) Our experimental study discussed later provides results 
consistent with this theoretical analysis. Moreover, our experiments also
show that 
the phase transition region is where the hard formulas are located, and
that hardness depends significantly on $t$.

We also considered the multi-component model $t$-$Q(a,e;A,E;m)$ of QBFs,
with the Chen-Interian model as its single-component specialization.
Let $q_t(a,e;A,E;m)$ be the probability that a random QBF 
from $t$-$\qbf(a,e;A,E;m)$ is true (in particular, $q_1(a,e;A,E;m)=$
$q(a,e;A,E;m)$). 
Using Theorem \ref{ci1} and reasoning as above, we can prove that
the phase transition regions for different values of $t$ coincide (and 
coincide with the phase transition region in the Chen-Interian model).  

\begin{thm}
\label{thm:4qbf}
For every integer $t\geq 1$ and real $r > 0$, if 
$\nu < \nu_l(a,e;r)$, then
{$\lim_{n\rra\infty} q_t(a,e;A,E;\lfloor \nu n\rfloor) =1$,}
and if $\nu > \nu_u(a,e;r)$, 
$\lim_{n\rra\infty} q_t(a,e;A,E;\lfloor \nu n\rfloor) =0$
(where $A=\lfloor r E\rfloor$ and $n=A+E$).
\end{thm}
\begin{proof}
For the proof, we will assume that the model $t$-$Q(a,e;A,E;m)$ generates 
QBFs with matrices that are ordered $t$-tuples of formulas generated from 
the model $C(a,e;A,E;n)$. As before, we have that for each fixed
positive integer $t$,
\begin{equation}
\label{eq:qbf}
q_t(a,e;A,E;\lfloor \nu n\rfloor) = 1 - (1 - q(a,e;A,E;\lfloor \nu n\rfloor))^t.
\end{equation}
This identity implies the claim in the same way as (\ref{eq:prob}) implied the 
assertion of Theorem \ref{thm:4}.
\nop{START NOP 5
Let $\nu < \nu_l(a,e;r)$.
Then, by definition of $\nu_l(a,e;r)$, $\lim_{n\rra\infty} q(a,e;A,E;\lfloor \nu n\rfloor) =1$.
Therefore, $\lim_{n\rra\infty} q_t(a,e;A,E;$ $\lfloor \nu n\rfloor) = $
$\lim_{n\rra\infty} \big( 1 - (1 - q(a,e;A,E;\lfloor \nu n\rfloor))^t \big) =$
$ \big( 1 - (1 - \lim_{n\rra\infty} q(a,e;A,E;\lfloor \nu n\rfloor))^t \big) = $
$ 1 - (1-1)^t = 1$.
Let $\nu > \nu_u(a,e;r)$.
Then, by definition of $\nu_u(a,e;r)$, $\lim_{n\rra\infty} q(a,e;A,E;\lfloor \nu n\rfloor) =0$.
Therefore, 
$\lim_{n\rra\infty} $ $q_t(a,e;A,E;\lfloor \nu n\rfloor) = $
$\big( 1 - (1 - \lim_{n\rra\infty} q(a,e;A,E;\lfloor \nu n\rfloor))^t \big) =  $
$1 - (1-0)^t = 0$. END NOP 5}
\end{proof}

The experimental results on satisfiability of QBFs from $t$-$Q(a,e;A,E;m)$,
which we present in Section \ref{sec:experiments}, agree with our theoretical analysis; we will also see 
there the easy-hard-easy 
pattern and a strong dependence of hardness on $t$. 

Finally, we considered the multi-component model $t$-$Q^\ctd(k,A,E)$, which 
incorporates both ideas we proposed in the paper. As in the other two cases, 
it is easy to derive the existence of the phase transition region and its 
invariance with respect to $t$ from the results on the underlying 
single-component model which, for the controlled model are given in
Theorem \ref{th:controlledThreshold}. Let $q_t^\ctd(k,A,E)$ denote the 
probability that a random formula in $t$-$Q^{\ctd}(k,A,E)$ is true.

\begin{thm}
\label{thm:t-ctd}
For every integer $t\geq 1$, if $\rho < \mu_l^\ctd(k)$, then
$\lim_{E\rra\infty} q_t^\ctd(k,\lfloor \rho E\rfloor,E)=1$,
and if $\rho > \mu_u^\ctd(k)$,
$\lim_{E\rra\infty} q_t^\ctd(k,\lfloor \rho E\rfloor,E)=0$,
\end{thm}
\begin{proof}
For the proof, we will assume that the model $t$-$Q^\ctd(k,\lfloor \rho E\rfloor,E)$ generates
QBFs with matrices that are ordered $t$-tuples of formulas generated from
the model $C^\ctd(k,\lfloor \rho E\rfloor,E)$ (disjunctions of $t$ formulas 
from the model, where repetitions of disjuncts are allowed and the order 
matters). As in the two classes of multi-component models we considered above,
we have that for each fixed positive integer $t$,
\begin{equation}
\label{eq:qbf}
q_t^\ctd(k,\lfloor \rho E\rfloor,E) = 1 - (1 - q^\ctd(k,\lfloor \rho E\rfloor,E))^t.
\end{equation}
This identity, when combined with Theorem \ref{th:controlledThreshold}, 
implies the assertion.
\end{proof}

We also studied the model $t$-$Q^\ctd(k,A,E)$ experimentally. The results are 
reported in Section \ref{sec:experiments}. As in other cases, they agree with
the predictions of the theoretical anaylysis above. Importantly, they show
that formulas from the model $t$-$Q^\ctd(k,A,E)$ can be ``super-hard.'' That
is, using ``controlled form'' CNF formulas in the disjunctions of a 
multi-component model, yields a way to generate formulas that are much 
harder than those generated from any other model considered before.

\section{Random Disjunctive Programs}\label{sec:encoding}

Our results on QBFs imply models of random disjunctive logic programs. This
is important as disjunctive logic programs increase the expressive power of 
answer set programming posing, at the same time, a computational challenge
\cite{BrewkaET11,KRHBook08}. 

Our approach to design models of random disjunctive programs is based on 
the translation from QBFs to programs due to Eiter and Gottlob 
\cite{EiterG95}. The Eiter-Gottlob translation works on QBFs $\Phi=
\exists X\forall Y G$, where $G$ is a DNF formula. 

To describe the translation, let us assume that $X=\{x_1,\ldots,x_E\}$,
$Y=\{y_1,\ldots,y_A\}$ and $G=D_1 \lor\ldots \lor D_m$, where 
\textcolor{black}{$D_i= L_{i,1} \land\ldots \land L_{i,k_i}$}
and $L_{i,j}$ are literals over $X\cup Y$.
For every atom $z\in X\cup Y$ we introduce a fresh atom $z'$. For every
$z\in X\cup Y$, we set $\sigma(z)=z$ and $\sigma(\neg z)=z'$. Finally,
we introduce one more fresh atom, say $w$, and define a disjunctive logic
program $P_\Phi$ to consist of the following rules:

\begin{tabbing}
xxx\=xxxxxxxxxxxxxxxxxxxxxxxxxx\= \kill
\>$z\lor z'$\>for each $z\in X\cup Y$\\
\>$y\lla w\ \ \mbox{and}\ \ y'\lla w$\>for each $y\in Y$\\
\>\textcolor{black}{$w \lla \sigma(L_{i,1}),\ldots,\sigma(L_{i,k_i})$}
\>for each $D_i$, $i=1,\ldots,m$\\
\>$w \lla \n w$
\end{tabbing}

\begin{thm}[Eiter and Gottlob \cite{EiterG95}]
Let $\Phi$ be a QBF $\exists X\forall Y G$, where $G$ is a DNF formula over
$X\cup Y$. Then $\Phi$ is true if and only if $P_\Phi$ has an answer set.
\end{thm}

We will use this result to derive models of disjunctive logic programs 
from the models of QBFs that we considered above. We recall that these 
models consist of formulas of the form $\forall X \exists Y F$, where 
$F$ is a CNF formula. Before we can apply the Eiter-Gottlob translation,
we have to transform these models (their formulas) into their dual 
counterparts. 

To this end, for a CNF formula $F$, we denote by $\overline{F}$ the formula
obtained from $\neg F$ by applying the De Morgan laws (thus, transforming 
$\neg F$ into DNF). Extending the notation, for each QBF $\Phi =\forall X
\exists Y F$, where $F$ is a CNF formula, we write $\overline{\Phi}$ for 
the QBF $\exists X\forall Y \overline{F}$. Clearly, $\Phi$ is true if and 
only if $\overline{\Phi}$ is false (or equivalently, $\Phi$ is false if 
and only if $\overline{\Phi}$ is true). 

\textcolor{black}{
\begin{corollary}
\label{cor:added}
Let $\Phi$ be a QBF $\forall X\exists Y G$, where $G$ is a CNF formula over
$X\cup Y$. Then $\Phi$ is false if and only if $P_{\overline{\Phi}}$ has 
an answer set.
\end{corollary}
}

Given a model (set) of QBFs of the form $\forall X \exists Y F$, where $F$
is a CNF formula, the mapping $\Phi \mapsto \overline{\Phi}$ transforms
the model into its dual, consisting of QBFs with a DNF formula in the matrix.
To these formulas we can apply the Eiter-Gottlob translation, thus 
obtaining a model (set) of disjunctive logic programs. 
\textcolor{black}{By Corollary \ref{cor:added}, this} model has the
same satisfiability properties as the original QBF model modulo
the switch between true and false.   


\textcolor{black}{
We now define $\overline{Q}(e,a;E,A;m)=\{\overline{\Phi}\colon\Phi\in 
Q(e,a;E,A;m)\}$. The model (set) $\overline{Q}(e,a;E,A;m)$ is the dual
to the Chen-Interian model $Q(e,a;E,A;m)$. Applying the Eiter-Gottlob 
translation $\Psi\mapsto P_\Psi$ to QBFs $\Psi \in \overline{Q}(e,a;E,A;m)$,
yields a model (set) of disjunctive logic programs, which we denote by 
$\dnf_{\dlp}(e,a;E,A;m)$. It follows from our comments after 
Corollary \ref{cor:added}
that the theoretical results we obtained for the Chen-Interian model 
$Q(e,a;E,A;m)$ apply directly to the model $\dnf_{\dlp}(e,a;E,A;m)$ 
(modulo the switch between true and false). }

%
%


\textcolor{black}{Next, we define $\overline{Q^\ctd}(k,E,A) =
\{\overline{\Phi}\colon \Phi\in
Q^{\ctd}(k,E,A)\}$. The model $\overline{Q^\ctd}(k,E,A)$ is dual to
our controlled model of QBFs. By applying the Gottlob-Eiter translation 
to QBFs in $\overline{Q^\ctd}(k,E,A)$, we obtain the model (set) of 
disjunctive logic programs, which we denote by $\dnf^\ctd_\dlp(k,E,A)$.
As before, by our comments following Corollary \ref{cor:added},
the models $Q^{\ctd}(k,E,A)$ and $\dnf^\ctd_\dlp(k,E,A)$ have the same 
satisfiability properties (modulo the switch between true and false).}
%

\subsection{Multi-component models of disjunctive logic programs} 

The translation proposed by Eiter and Gottlob can be extended to QBFs of 
the form $\Phi= \exists X \forall Y G$, where $G =G_1\land\ldots\land G_t$
and each $G_i$ is a DNF formula. The translation is similar, except that we 
need $t$ additional variables $w_1,\ldots,w_t$ to represent DNF formulas 
$G_i$. The translation consists of rules 
\begin{tabbing}
xxx\=xxxxxxxxxxxxxxxxxxxxxxxxxx\= \kill
\>$z \lor z'$\>for each $z\in Z$\\
\>$y \lla w$\ \ \mbox{and}\ \ $y' \lla w$\>for each $y\in Y$\\
\>$w \lla w_1, \ldots, w_t$\ \ \mbox{and}\ \ $w\lla \n w$
\end{tabbing}
that form the \emph{fixed} part of the translation, and its \emph{core}
consisting of Horn rules 
\begin{tabbing}
xxx\=xxxxxxxxxxxxxxxxxx\= \kill
\>$w_h \lla z_1,\ldots,z_\ell$
\end{tabbing}
where $h=1,\ldots,t$, and the rules with the head $w_h$
are obtained from the formula $G_h$ just as in the original Eiter-Gottlob
translation (except that $w_h$ is now used as the head and not $w$). 
In fact, in the case when $t=1$ the program above coincides with the 
result of the Eitr-Gottlob translation modulo a rewriting, in which we 
eliminate the rule $w\lla w_1$ and replace $w_1$ in the head of each rule 
in the core with $w$.

Extending the earlier notation, we denote the program described above by 
$P_\Phi$. The following result can be derived by an argument similar to 
that Eiter and Gottlob used to prove their theorem.

\begin{thm}
Let $\Phi = \exists X \forall Y (G_1\land\ldots\land G_t)$, where each $G_i$
is a DNF formula. Then $\Phi$ is true if and only if $P_\Phi$ has an answer 
set.
\end{thm}

We can now derive multi-component models of disjunctive logic programs 
from the multicomponent models of QBFs. The basic idea is the same as before.
A multi-component model of QBFs gives rise to its dual via a transformation
$\Phi \mapsto \overline{\Phi}$ (it consists of negating $\Phi$ and applying
De Morgan laws). Next, the translation above transforms QBFs from that dual 
model into disjunctive programs, yielding the corresponding multi-component
model of programs. We apply this approach to two multi-component models of
QBFs we considered in this paper: $t$-$Q(e,a;E,A;m)$ and $t$-$Q^\ctd(k,E,A)$.
We denote the corresponding models of disjunctive logic programs by 
$t$-$\dnf_\dlp(e,a;E,A;m)$ and $t$-$\dnf_\dlp^\ctd(k,E,A)$. 


\begin{corollary}
\label{cor:added2}
\textcolor{black}{Let $\Phi = \exists X\forall Y F$, where $F\in \mbox{$t$-$Q(e,a;E,A;,m)$}$
or $F\in \mbox{$t$-$Q^\ctd(k,E,A)$}$. Then, $\Phi$ is false ($\overline{\Phi}$ 
is true) if and only if $P_{\overline{\Phi}}$ has an answer set. }
\end{corollary}

\textcolor{black}{By Corollary \ref{cor:added2}, the models $t$-$Q(e,a;E,A;,m)$ 
($t$-$Q^\ctd(k,E,A)$, respectively) and $t$-$\dnf_\dlp(e,a;E,A;m)$ 
($t$-$\dnf_\dlp^\ctd(k,E,A)$, respectively) have the same satisfiability 
properties (modulo the switch between true and false).}
\section{Empirical analysis}\label{sec:experiments}

We now describe an experimental analysis of the behavior of our models 
and discuss their properties.

\subsection{Experiment Setup}

\newcommand{\myOneCol}{0.83\columnwidth}
\newcommand{\vsFig}{\vspace*{-1cm}}

\begin{figure*}[t!]
    \centering
    \captionsetup[sub]{font=normal,labelfont={sf,sf}} 
    \begin{subfigure}[t]{\columnwidth}    \centering 
        \includegraphics[page=1,width=\myOneCol]{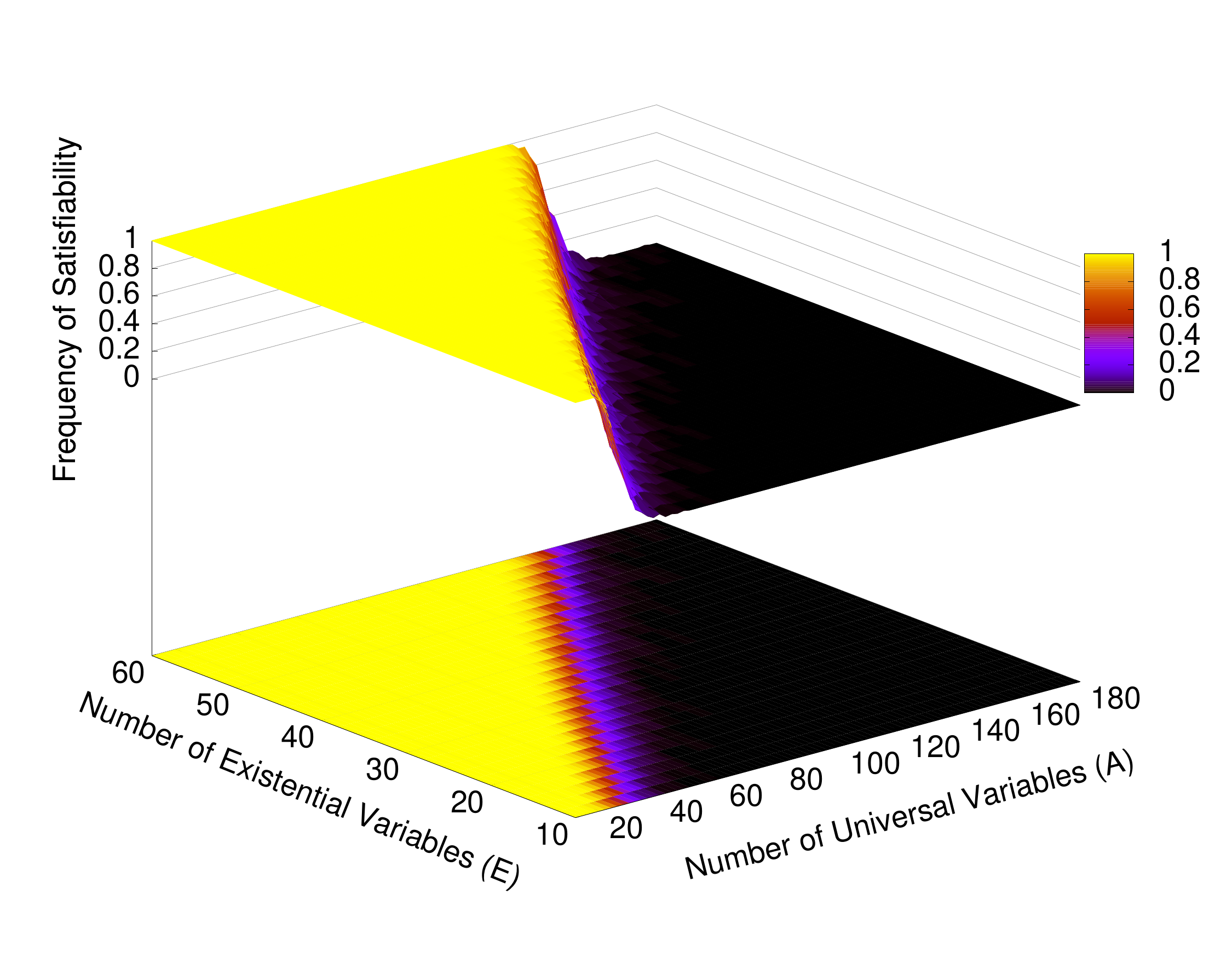}
    \caption{Phase transition (Controlled)}
	\label{fig:CTRL:Sat}
	\end{subfigure}
\vspace*{0.5cm}
    \begin{subfigure}[t]{\columnwidth}    \centering 
        \includegraphics[page=2,width=\myOneCol]{ctrl-s-128-E-10-300-2-A-10-60-2-w-1-1-1-3D-plane.pdf}
    \caption{Hardness (Controlled)}
	\label{fig:CTRL:Time} 
	\end{subfigure}

%
%
    \caption{Behavior of controlled model: phase transition and hardness.}\label{fig:CTRL}
\end{figure*}

\renewcommand{\myOneCol}{0.93\columnwidth}
\begin{figure*}[t!] 
    \centering
    \captionsetup[sub]{font=normal,labelfont={sf,sf}} 

    \begin{subfigure}[t]{\columnwidth}    \centering 
        \includegraphics[width=\myOneCol]{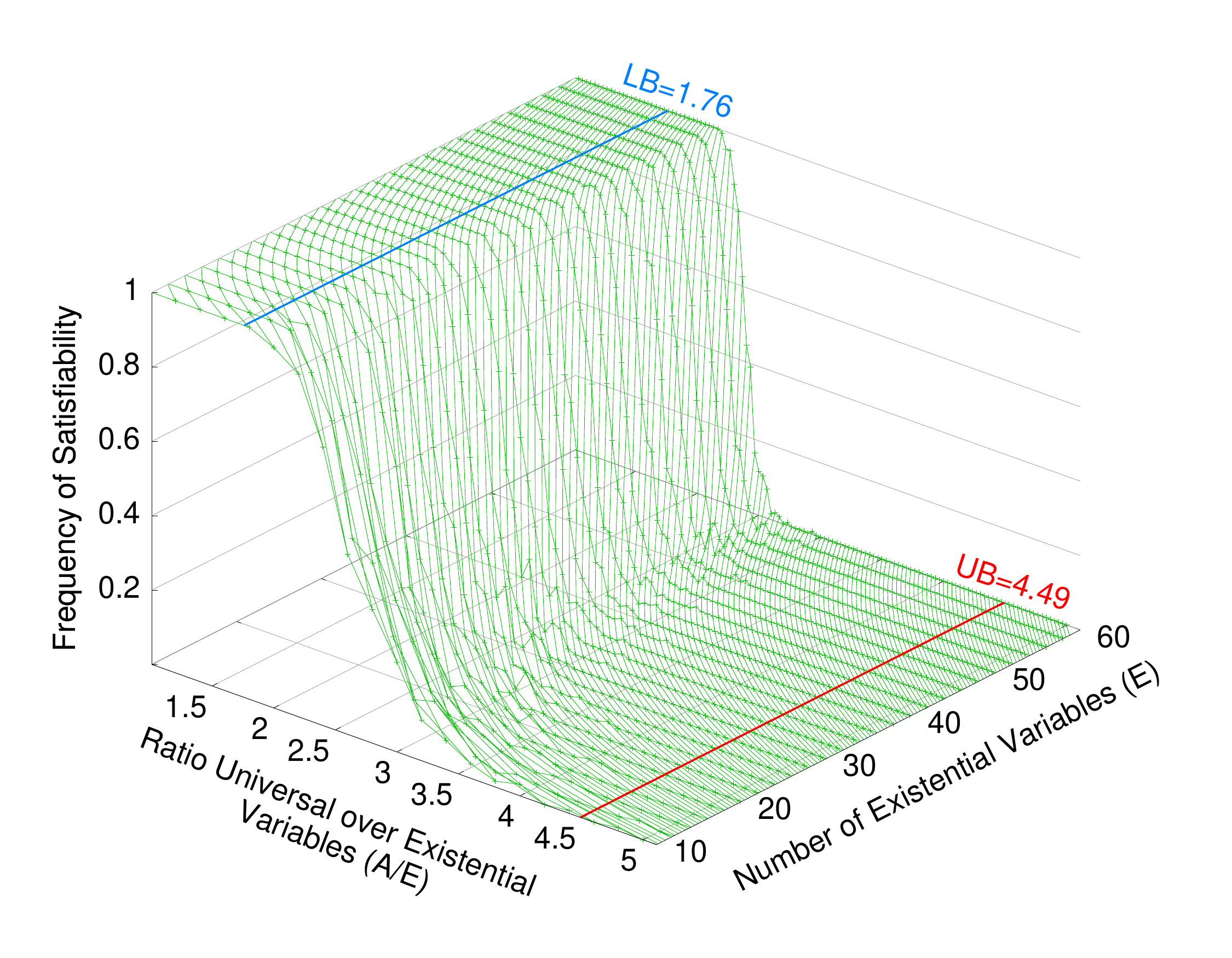}
    \caption{Bounds on Satisfiability (Controlled)}\label{fig:CTRL:Bounds}
	\end{subfigure}    
	\vspace*{0.5cm}
	
    \begin{subfigure}[t]{\columnwidth}    \centering 
        \includegraphics[width=\myOneCol]{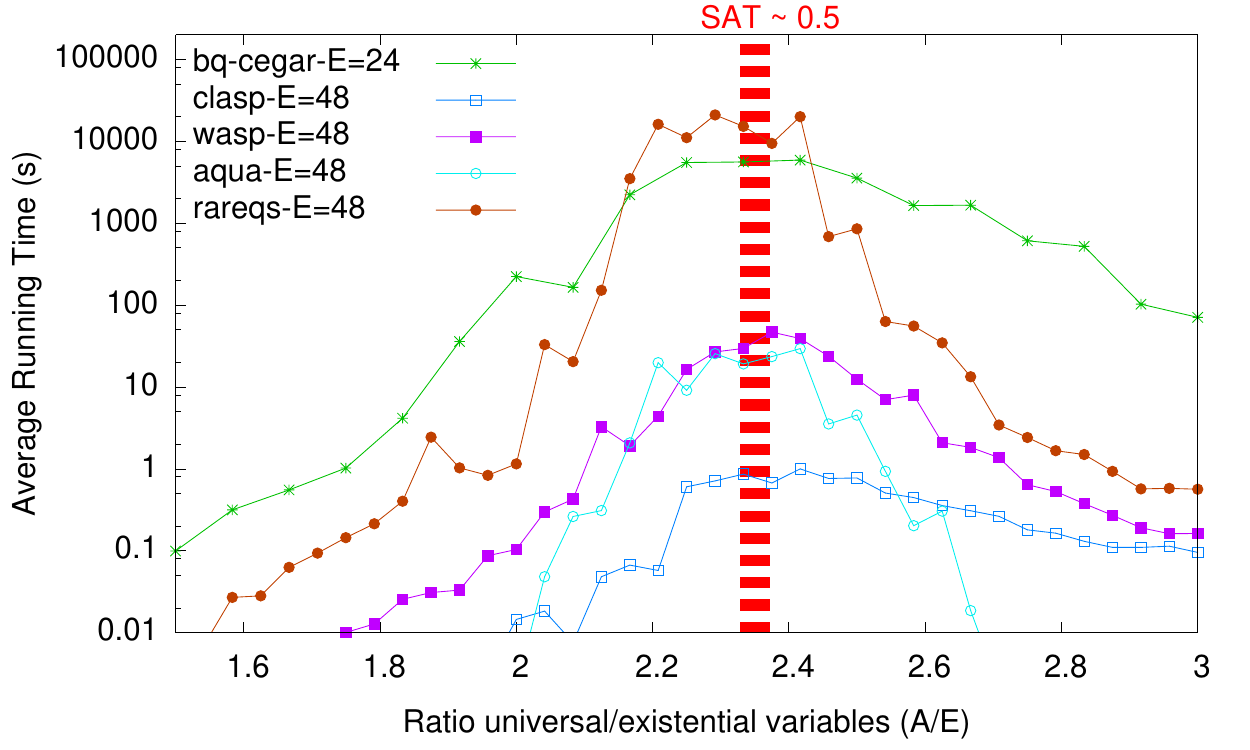}
    \caption{Solver Independence (Controlled)}\label{fig:CTRL:Solvers} 
	\end{subfigure}    
    
%
    \caption{Behavior of controlled model: bounds and solver independence.}\label{fig:CTRL2}
\end{figure*}

To claim that properties and patterns are inherent to a model and not an
artifact of a solver used, we performed our experiments with several 
well-known SAT, QBF and ASP solvers. The SAT solvers included
\glucose 4.0~\cite{AudemardLS13};
\lingeling, version of 2015~\cite{Biere14};
and \kcnfs, version of SAT'07 competition~\cite{DBLP:journals/jar/DequenD06}.
The QBF solvers included \bqcegar (a combination of \textit{bloqqer} 
preprocessor~\cite{HeuleJLSB15} and
\textit{ghostq}~\cite{KlieberJMC13} solver from QBF gallery 2014);
\aigsolve~\cite{PigorschS10};
\rareqs~\cite{DBLP:journals/ai/JanotaKMC16}, version 1.2 from QBF competition 2016;
and \aqua \footnote{\url{www.qbflib.org/DESCRIPTIONS/aqua16.pdf}}.
Finally, the two ASP solvers we used in experiments were
\clasp 3.1.3~\cite{GebserKNS07a}
and \wasp 2.1~\cite{AlvianoDLR15}, both paired with
\textit{gringo} 4.5.3~\cite{GebserKKS11}. 
All solvers were run in their default configurations.
We stress that we did not aim at comparing solver performance, instead
\textit{our goal was to identify solver-independent properties inherent to a model}.

To support experiments, we developed a tool in Java to generate random
CNF formulas from $\cnf(k,n,m)$, QBFs from $Q(a,e;A,E;m)$ and $Q^\ctd(k,A,E)$, 
and programs from $\dnf_\dlp(e,a;E,A;m)$ and $\dnf_\dlp(k,E,A)$ (``dual''
to QBFs from $Q(e,a;E,A;m)$ and $Q^\ctd(k,E,A)$). For each
class $\mathcal{C}$ of formulas and programs listed, our tool generates 
also formulas (programs) from the corresponding multicomponent model
$t$-$\mathcal{C}$. 

\textcolor{black}{
Formulas and QBFs generated according to the multi-component models 
$t$-$\cnf(k,n,m)$, $t$-$Q(a,e;A,E;m)$ and $t$-$Q^\ctd(k,A,E)$, where 
$t>1$, are non-clausal or have non-clausal matrices (in the case of QBFs). 
As they do not adhere to the (Q)DIMACS format required by SAT/QBF solvers, the generator
transforms non-clausal formulas to CNF using the Tseitin transformation
\cite{Tseitin1983}.
That transformation introduces fresh auxiliary variables (while replacing binary subformulas) and new clauses (modeling the equivalence of each replacement) to obtain a CNF formula that is equisatisfiable to the original one.
The Tseiting transformation is efficient, since it only causes a linear growth in size (whereas doing the same normalization via distributivity laws may lead to an exponential blow-up).%
\footnote{For this reason the Tseitin transformation is employed very often in real-world applications of SAT/QBF. Actually, many formulas used in SAT and QBF competitions~\cite{JarvisaloBRS12,NarizzanoPT06} come from applying it to non-normal form inputs suggested by problem statements.}
Interestingly, the logic programs in the models $t$-$\dnf_\dlp(e,a;E,A;m)$ and 
$t$-$\dnf_\dlp(k,E,A)$ have a much simpler structure than the corresponding
Tseitin-transformed formulas from the ``dual'' models $t$-$Q(e,a;E,A;m)$ and
$t$-$Q^\ctd(k,E,A)$). As can be seen from the translation, these programs need 
new variables only to represent each of the $t$ components (disjuncts) of 
the matrix formula.
}

Once a formula $\Phi$ is generated, it is stored in two files:
one with an encoding of $\Phi$ in the (Q)DIMACS numeric
format of (Q)SAT solvers~\cite{JarvisaloBRS12,NarizzanoPT06}, and the 
other one with the disjunctive logic program corresponding to $\Phi$
in the ASPCore 2.0 syntax~\cite{CalimeriGMR16}. 
As discussed in the previous section, since the programs are generated 
from the negations of the QBFs in our random QBF models, they have answer
sets if and only if the original formulas are false. Thus, when we analyze 
satisfiability we plot only the curves obtained by evaluating either the 
formulas or the corresponding logic programs (the plots are symmetric to 
each other). In all the experiments the results are averaged over 128 
samples of the same size.

Experiments were run on a Debian Linux with 2.30GHz Intel Xeon E5-4610 v2 CPUs and 128GB of RAM.
Each execution was constrained to one single core by using the \textit{taskset} command.
Time measurements were performed by using the \textit{runlim} tool.
The generator used in the experiments is publicly available 
at~\url{https://www.mat.unical.it/ricca/RandomLogicProgramGenerator}.

\renewcommand{\myOneCol}{0.83\columnwidth}
\begin{figure}[t!]
    \centering  
    \captionsetup[sub]{font=normal,labelfont={sf,sf}} 

    \begin{subfigure}[t]{\columnwidth}    \centering 
        \includegraphics[page=1,width=\myOneCol]{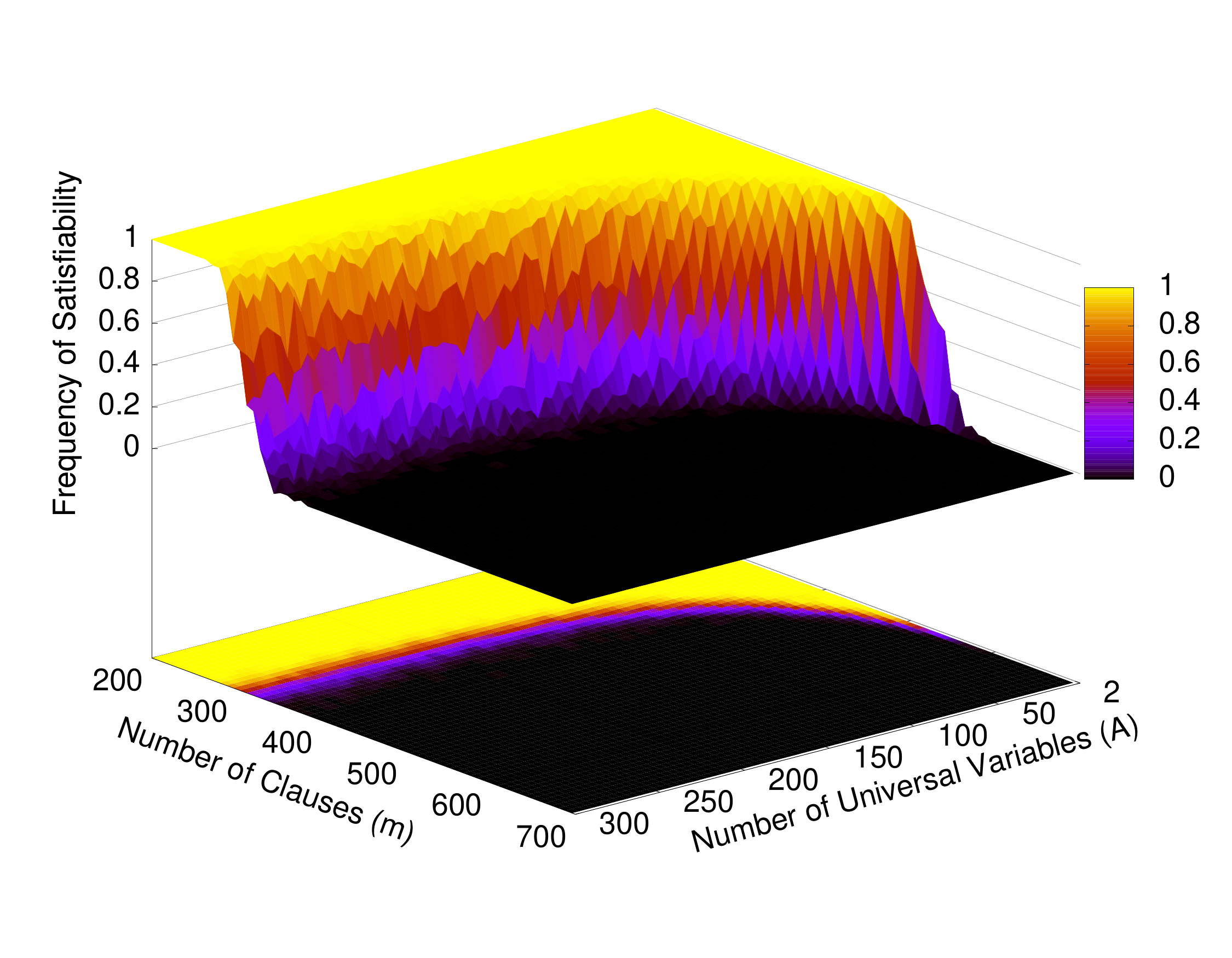}
    \caption{Phase transition (Chen-Interian)}\label{fig:CIvsCTRL:CISat}
	\end{subfigure}
	\vspace*{0.5cm}    
    
    \begin{subfigure}[t]{\columnwidth}    \centering 
    		\includegraphics[page=2,width=\myOneCol]{ci-s-64-E-4-300-4-A-70-70-1-c-200-700-8-w-1-1-1-3D.pdf}
    \caption{Hardness (Chen-Interian)}\label{fig:CIvsCTRL:CITime}
	\end{subfigure}


    \caption{Chen-Interian: Phase transition and Hardness.}\label{fig:CIvsCTRL:CI}
\end{figure}

\begin{figure}[t!]
    \centering  
    \includegraphics[width=\myOneCol]{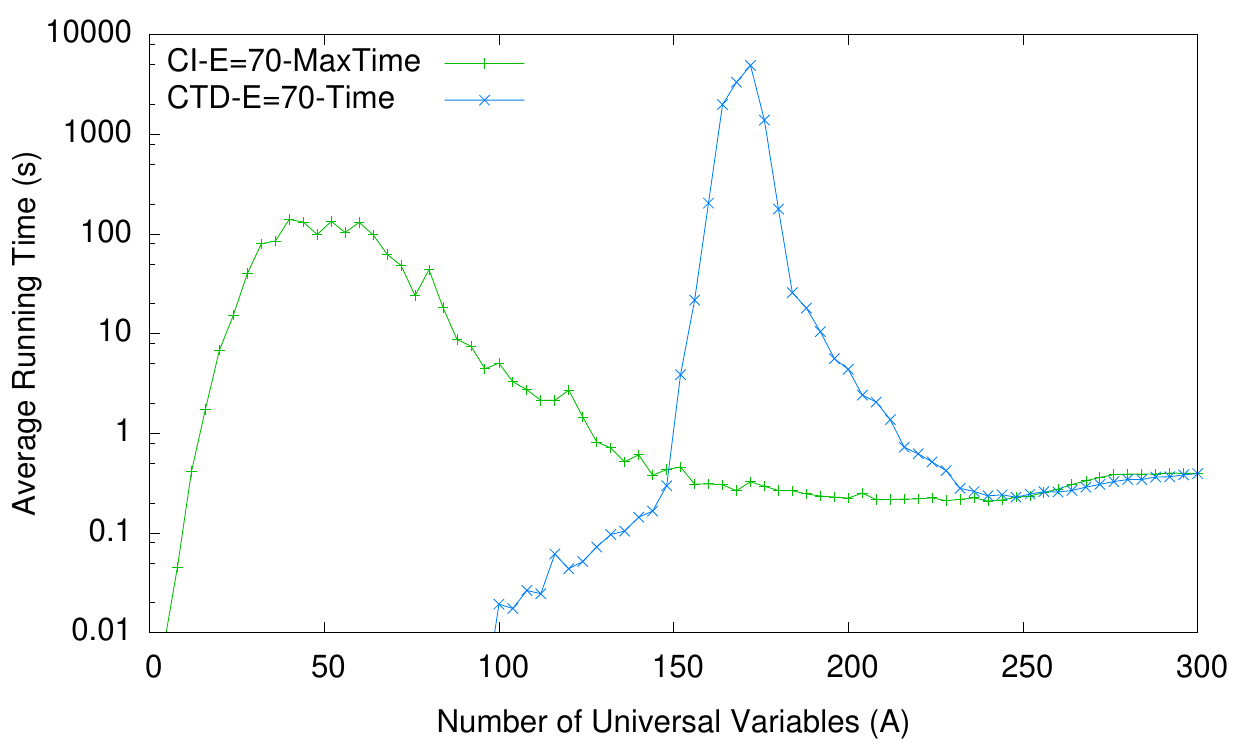}
    \caption{Comparing Chen-Interian and Controlled model: hardness comparison.}\label{fig:CIvsCTRL:Comparison}
\end{figure}

\subsection{Behavior of the controlled model}
We first study the satisfiability and hardness of formulas and corresponding 
programs generated according to the controlled model.
We generated QBF instances from the model $Q^{\ctd}(4,A,E)$ and program 
instances from the dual model $\dnf_\dlp(4,A,E)$) for the parameters $E$ 
and $A$ ranging over $[10..60]$ and $[20..180]$, respectively (consequently,
the number of clauses ranges from $40$ to $360$).

Figure~\ref{fig:CTRL:Sat} shows the satisfiability results for the model 
$Q^{\ctd}(4,A,E)$. The picture for $\dnf_\dlp(4,A,E)$ is dual (symmetric 
with respect to the plane given by the frequency of satisfiability equal 
to $1/2$); the results we show were in fact obtained by running \clasp 
on programs from $\dnf_\dlp(4,A,E)$ and adapted to the case of 
$Q^{\ctd}(4,A,E)$). 
The gradient of colors ranging from yellow (QBF true) 
to black (QBF false) helps to identify the phase transition region, which 
is also projected on the $A$-$E$ plane below. We observe that phase transitions occur for a specific value of the ratio 
between universal and existential variables, specifically, for $A/E\simeq 
2.37$. A different perspective on the same data is presented in 
Figure~\ref{fig:CTRL:Bounds}, where the frequency of satisfiability is
depicted with respect to the ratio $A/E$, and where the two straight lines 
show the bounds predicted by Theorem~\ref{th:controlledThreshold}, assuming 
the bounds for satisfiability and unsatisfiability of 3-CNF formulas~\cite{KaporisKL03,DiazKMP09}
(i.e., $\mu^\ctd_l(4)\geq \frac{3.52}{2}=1.76$ and $\mu^\ctd_u(4)\leq 4.49$).
We observe that the transition sharpens when the number of variables grows, and the transition occurs within the bounds predicted by the theoretical results.

\nop{
Figure~\ref{fig:CTRL:Sat} shows the satisfiability of formulas and corresponding
programs from the controlled model. We generated QBF instances from the
model $Q^{\ctd}(4,A,E)$ (and program instances from the dual model
$\dnf_\dlp(4,A,E)$)) by varying $E$ and $A$ over $[10..60]$ and $[20..180]$
(consequently, the number of clauses ranges from 40 to 360). The gradient of 
colors ranging from black (no answer set) to yellow (answer set found) helps to identify the phase
transition region, which is also projected on the $A$-$E$ plane below.
We observe that phase transitions occur for a specific value of the ratio
between universal and existential variables, specifically, for $A/E\simeq
2.37$. A different perspective on the same data is presented in
Figure~\ref{fig:CTRL:Bounds}, where the frequency of satisfiability is
depicted with respect to the ratio $A/E$, and where the two straight lines
show the bounds predicted by Theorem~\ref{th:controlledThreshold}, assuming
the bounds for satisfiability and unsatisfiability of 3-CNF formulas~\cite{KaporisKL03,DiazKMP09}
(i.e., $\mu^\ctd_l(4)\geq \frac{3.52}{2}=1.76$ and $\mu^\ctd_u(4)\leq 4.49$).
We observe that the transition sharpens when the number of variables grows, and the transition occurs within the bounds predicted by the theoretical results.
}

To study the hardness of formulas, the average running times are plotted in Figure~\ref{fig:CTRL:Time}.
Here the gradient of colors ranging from black (basically instantaneous execution) to yellow (the maximum average
running time) helps to identify the hardness region.
As before the region is also projected on the $A$-$E$ plane below.
As expected hardness arises around the phase transition region and grows 
with the number of variables.
To provide evidence that the hardness of the controlled model is independent of the solver used, 
Figure~\ref{fig:CTRL:Solvers} plots the average execution times when
running two QBF solvers (\rareqs and \aqua) and two ASP solvers (\clasp and \wasp) 
on formulas/programs implied by formulas
from $Q^{\ctd}(4,A,E)$ with 48 existential variables, and 
the QBF solver \bqcegar on formulas with 24 existential variables (for
that solver, we had to decrease the size of formulas to ensure termination
within a reasonable time). 
We note that all solvers find hard formulas in the same region, and the maximum hardness
coincides with the transition zone marked by the red vertical strip.
No data is reported in Figure~\ref{fig:CTRL:Solvers} for \aigsolve because it terminated abruptly in some instances 
(throwing \textit{std::bad\_alloc}) and in some other we had to kill the process after 15 days of execution. 
(This behavior is probably due to a memory access problem.)

\subsection{Controlled vs Chen-Interian model}
\label{cvci}
We now compare the controlled and the Chen-Interian models with respect to
the hardness of formulas having the same number of variables. 

We start by presenting results on the behavior of the 
Chen-Interian model $Q(a,e;A,E;m)$, where we set $a=1$, $e=3$, and $E=70$, 
and vary the number $A$ of universal variables over the range $[2..300]$
and the number $m$ of clauses over the range $[200..700]$.
These results are shown in Figure~\ref{fig:CIvsCTRL:CI}. They confirm and 
extend the findings by Chen and Interian \cite{ChenI05}. As before the 
gradient of colors in Figure~\ref{fig:CIvsCTRL:CI}, ranging from black to 
yellow, outlines the phase transition and the easy-hard-easy pattern. The 
surface is also projected onto the $A$-$m$ plane for an alternative 
visualization. For every value of $A$ 
(in fact, for every value of the ratio $A/E$; indeed, we recall that in 
our experiment the we fixed the value of $E$ to $70$), as we grow $m$ we 
observe the phase 
transition. The place where this phase transition occurs depends on $A$ 
(more generally, on the ratio $A/E$; but in our experiments $E$ is fixed). 
For each value of $A$ (more precisely, for each value of $A/E$), the 
hardest formulas are located around the phase transition area, as evidenced 
by Figure \ref{fig:CIvsCTRL:CI}(b). The behavior presents there only 
for the values of $A$ of up to about $85$; for higher values of $A$, the 
running times even on the formulas from the phase transition region are very
small. Figure \ref{fig:CIvsCTRL:CI}(b) also shows that the overall peak of
hardness occurs in the phase transition region for a specific value of $A$ 
or, as explained earlier, for a specific value of the ratio $A/E$.

Next, we compare the hardness properties of the controlled and the 
Chen-Interian models with the same number of existential variables, which 
can be viewed as a measure of the hardness of individual SAT instances that 
arise while solving a QBF of the form $\forall\exists F$. The graphs in 
Figure~\ref{fig:CIvsCTRL:Comparison} capture the behavior of the hardness 
for the two models under this constraint. For the controlled model, for 
each value of $A$, the value on the corresponding hardness graph (the blue 
line) is obtained by averaging the solve times on formulas generated from 
the model $Q^{\ctd}(4,A,70)$. The matrices of these formulas are 4-CNF 
formulas over $A+70$ variables and with $2A$ clauses. The corresponding 
point on the hardness graph for the Chen-Interian model is obtained by 
averaging the solve times on formulas generated from the model 
$Q(1,3;A,70;max)$, where for each $A$ (and $E=70$), $max$ is selected to 
maximize the solve times (in particular, $max$ falls in the phase transition 
region for the combination of the values $A$ and $E=70$). The matrices of
these formulas are 4-CNF formulas over $A+70$ variables and $max$ clauses.

The results show that the peak hardness regions for the
two models are not aligned. The hardest formulas over 70 existential
variables from the Chen-Interian models have $A\approx 50\mbox{-}55$
universal variables and $m=350$ clauses. The hardest formulas over 70
existential variables from the controlled model have $A\approx 170$ and
$m\approx 340$. Our results show that the hardest formulas from the
controlled model are almost two orders of magnitude harder than the hardest
formulas from the Chen-Interian model. On the other hand, while the hardest 
formulas (for a fixed value of $E$, here $E=70$) in the two models have 
similar numbers of clauses (about 340-350), the Chen-Interian model formulas
have fewer universal variables (about 50-55 versus 170 in the controlled model).

It is also useful to look at the point where the hardness of one model meets 
the other. It happens for $A\approx 150$. At this point, the CNF formulas 
that are the matrices of QBFs from the controlled model have 70 existential 
and about 150 universal variables, and about 300 clauses. The corresponding 
parameters for the formulas from the Chen-Interian model have very similar 
values. Indeed, the hardest formulas for the Chen-Interian model when $E=70$ 
and $A=150$ have about 300 clauses (cf. Figure \ref{fig:CIvsCTRL:CI}).

To summarize, a direct comparison for the hardness of the two models is not 
clear cut. On the one hand, our results show that if we make the comparison 
for models with the same number of existential variables the points, in terms 
of $A$, in which the two model generate their hardest instances are very 
different. On the other hand, there is a setting (corresponding to the phase 
transition for the controlled model) in which the controlled model generates 
much harder formulas than any other setting (corresponding to a phase 
transition) for the Chen-Interian model.

For the sake of completeness, we report that we obtained 
results consistent with those discussed above experimenting with other settings
of existential variables and clause lengths.

\renewcommand{\myOneCol}{0.70\columnwidth}
\begin{figure*}[t!]
    \centering
    \captionsetup[sub]{font=normal,labelfont={sf,sf}} 

    \begin{subfigure}[t]{\columnwidth}\centering
	\includegraphics[width=\myOneCol]{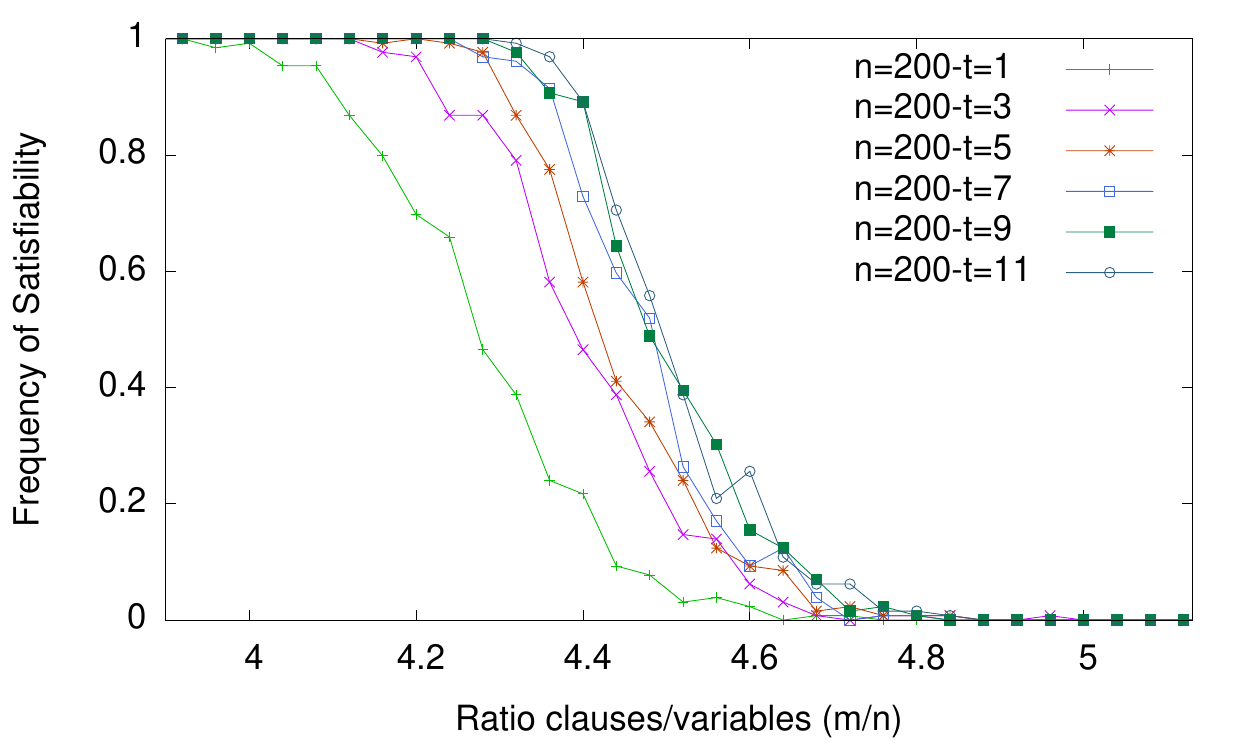}
    \caption{Phase transition shift (SAT)}\label{fig:growing:sat}
	\end{subfigure}
	
    \begin{subfigure}[t]{\columnwidth}\centering
    \includegraphics[width=\myOneCol]{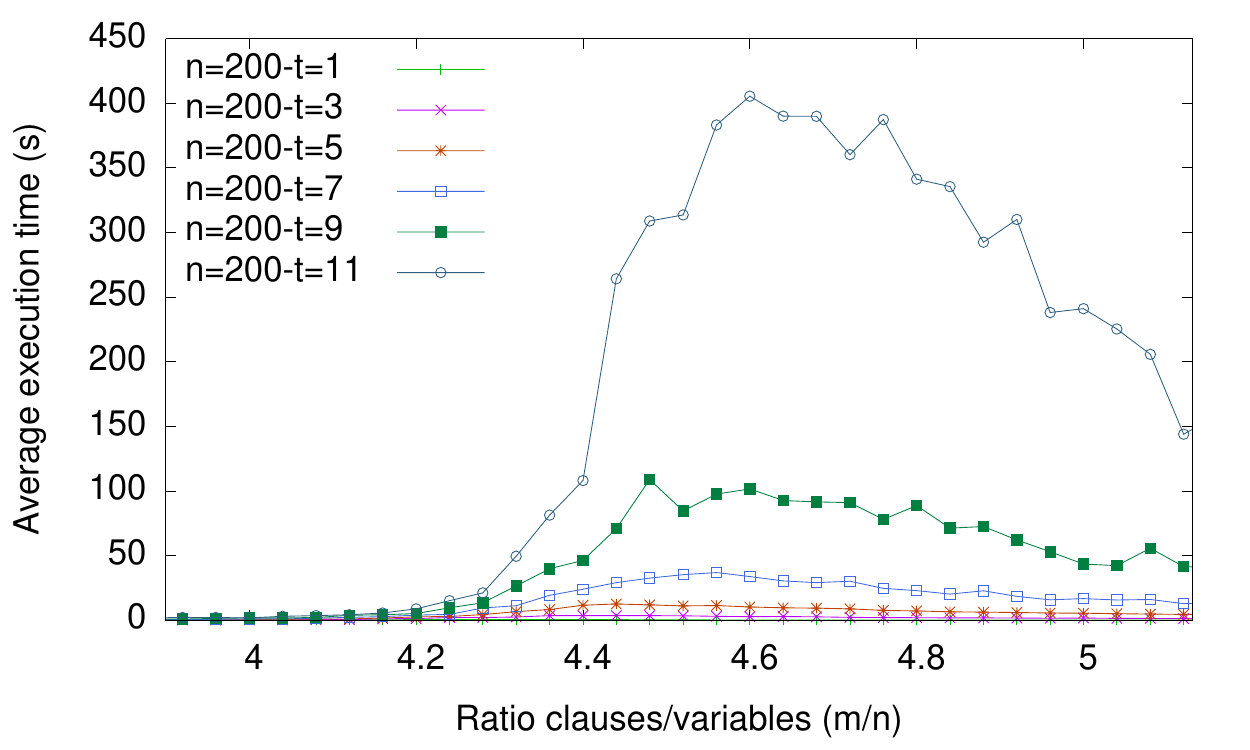}
    \caption{Easy-hard-easy pattern (SAT)}\label{fig:growing:onesolver:sat}
	\end{subfigure}
	
    \begin{subfigure}[t]{\columnwidth}\centering
    \includegraphics[width=\myOneCol]{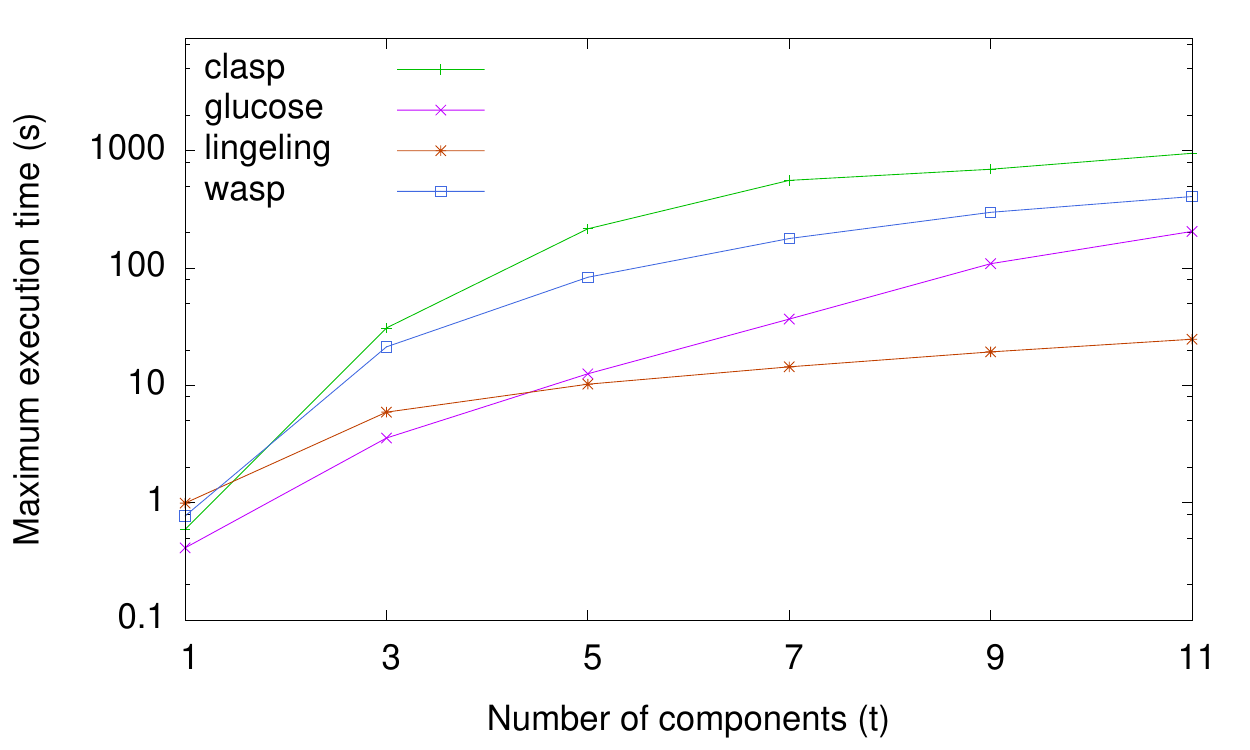} 
    \caption{Hardness of multi-component model (SAT)}\label{fig:growing:sat:maxtime}
	\end{subfigure}
	 
    \caption{Behavior of multi-component model (SAT): phase transition and hardness.}\label{fig:multi}
\end{figure*}

\begin{figure*}[t!]
    \centering
    \captionsetup[sub]{font=normal,labelfont={sf,sf}} 
    \begin{subfigure}[t]{\columnwidth}\centering
    \includegraphics[width=\myOneCol]{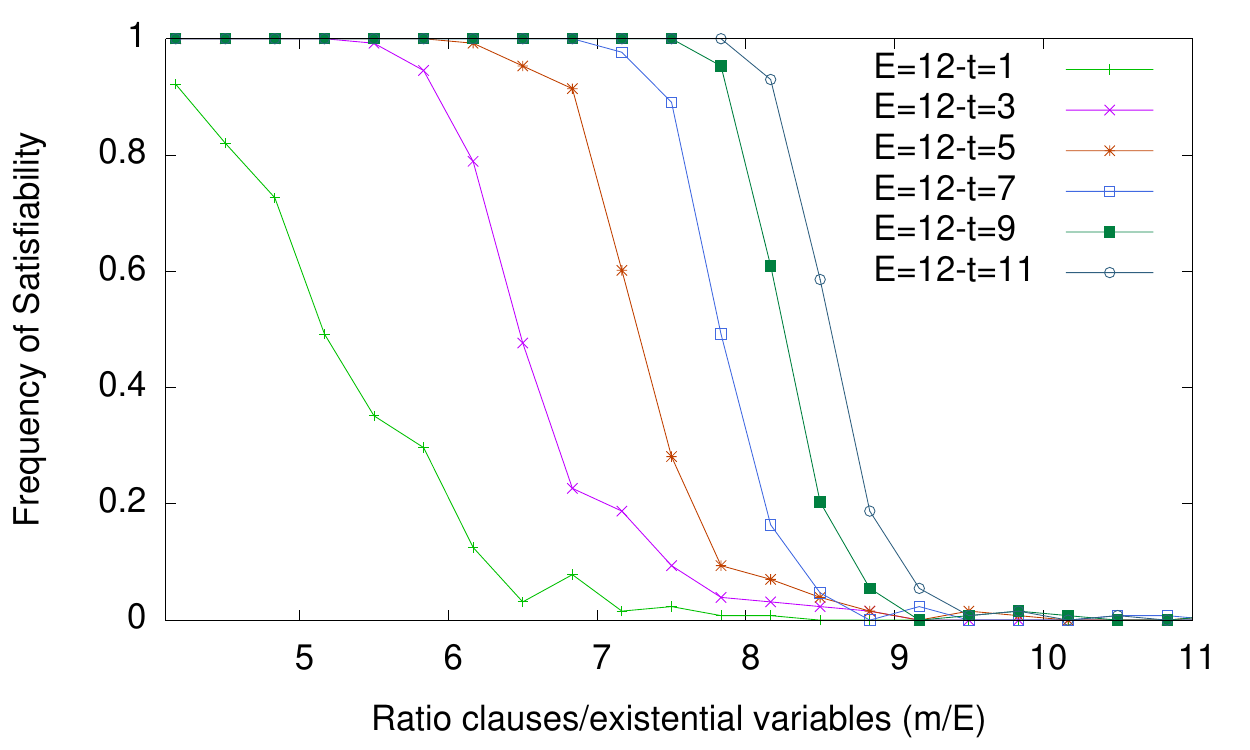} 		
    \caption{Phase transition shift (QBF)}\label{fig:growing:qbf}
	\end{subfigure}
	
    \begin{subfigure}[t]{\columnwidth}\centering
    \includegraphics[width=\myOneCol]{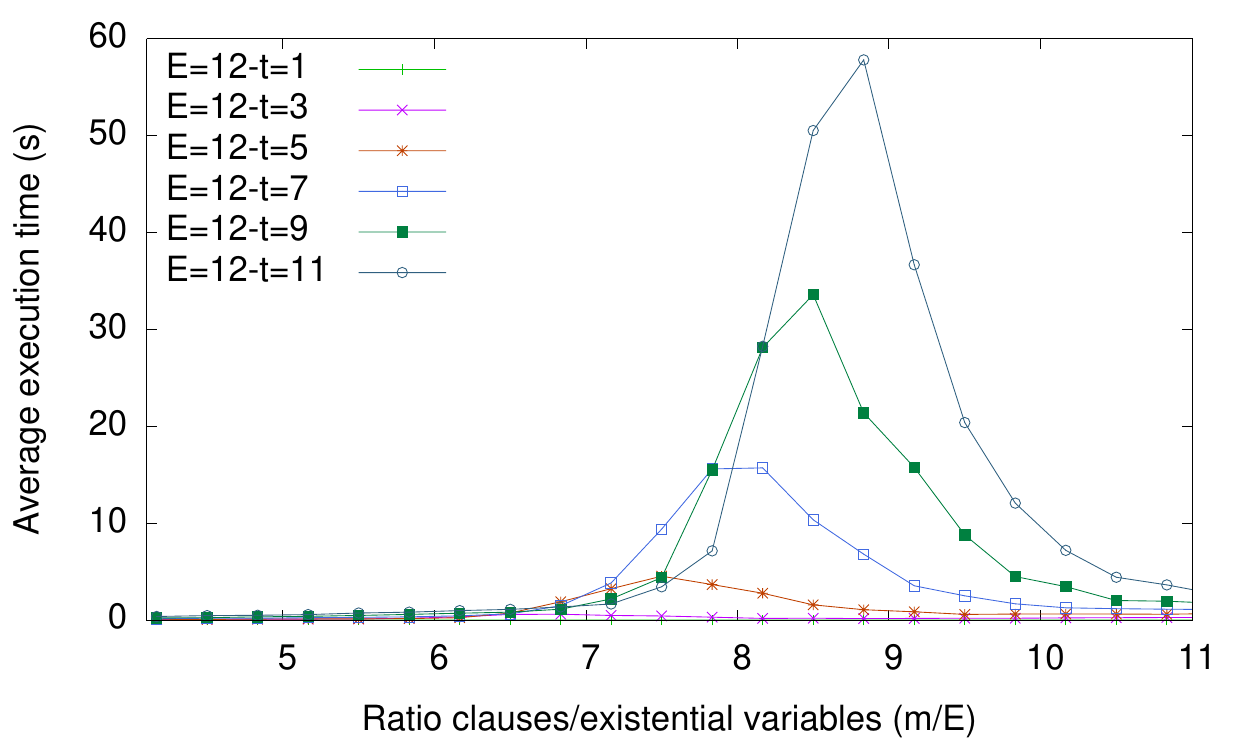} 
    \caption{Easy-hard-easy pattern (QBF)}\label{fig:growing:onesolver:qbf}
	\end{subfigure}
	
    \begin{subfigure}[t]{\columnwidth}\centering
    \includegraphics[width=\myOneCol]{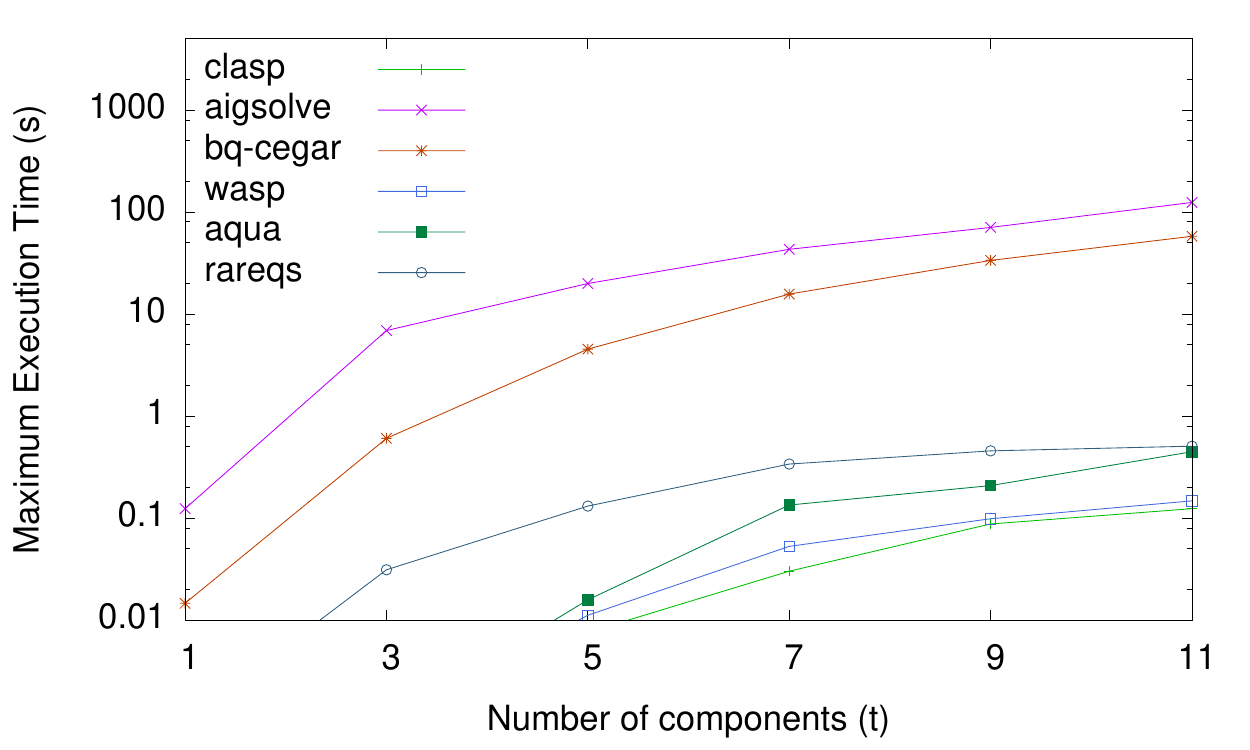} 	
    \caption{Hardness of multi-component model (QBF)}\label{fig:growing:qbf:maxtime}
	\end{subfigure}
    
    \caption{Behavior of multi-component model (QBF): phase transition and hardness.}\label{fig:multi}
\end{figure*}

\renewcommand{\myOneCol}{0.998\columnwidth}
\begin{figure*}[t!]
    \centering
    \captionsetup[sub]{font=normal,labelfont={sf,sf}} 
    	\vspace*{1.2cm}
    \begin{subfigure}[t]{\columnwidth}\centering
    \includegraphics[width=\myOneCol]{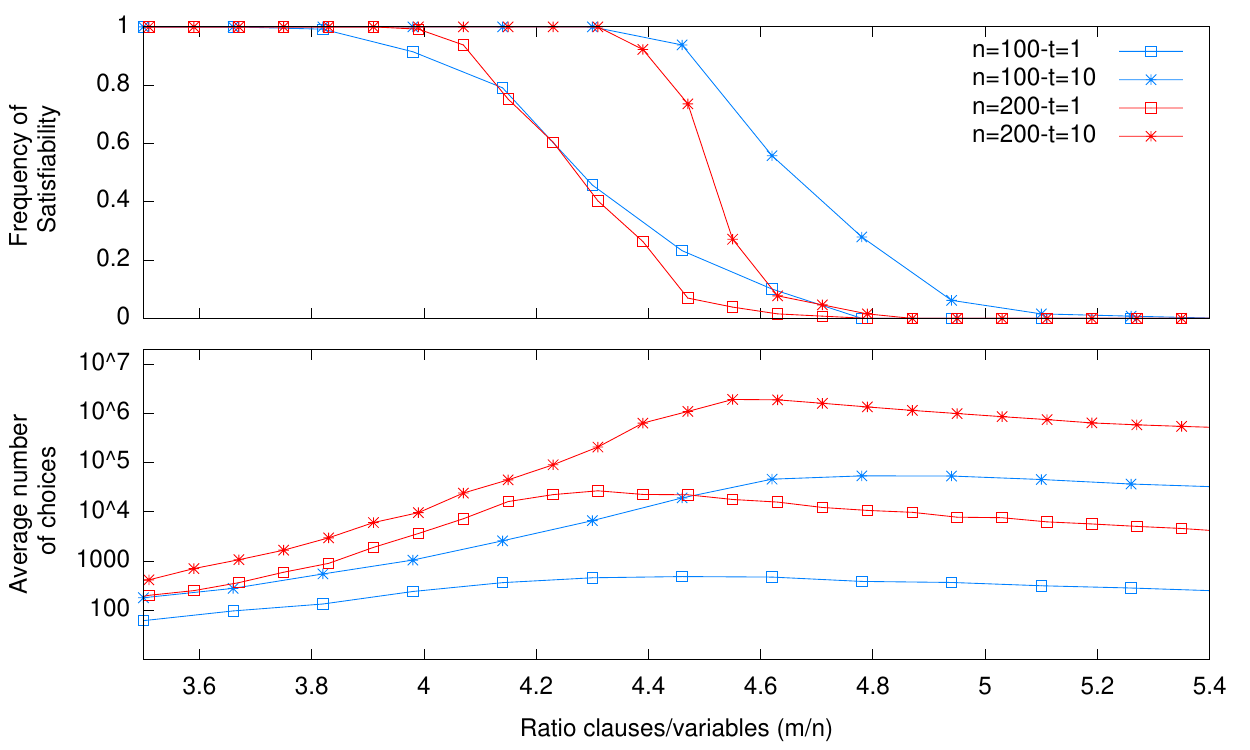} 	
    \caption{Combined effect of variables and components (SAT)}\label{fig:bounds:sat}
	\end{subfigure}
	\vspace*{0.8cm}
	
    \begin{subfigure}[t]{\columnwidth}\centering
	\includegraphics[width=\myOneCol]{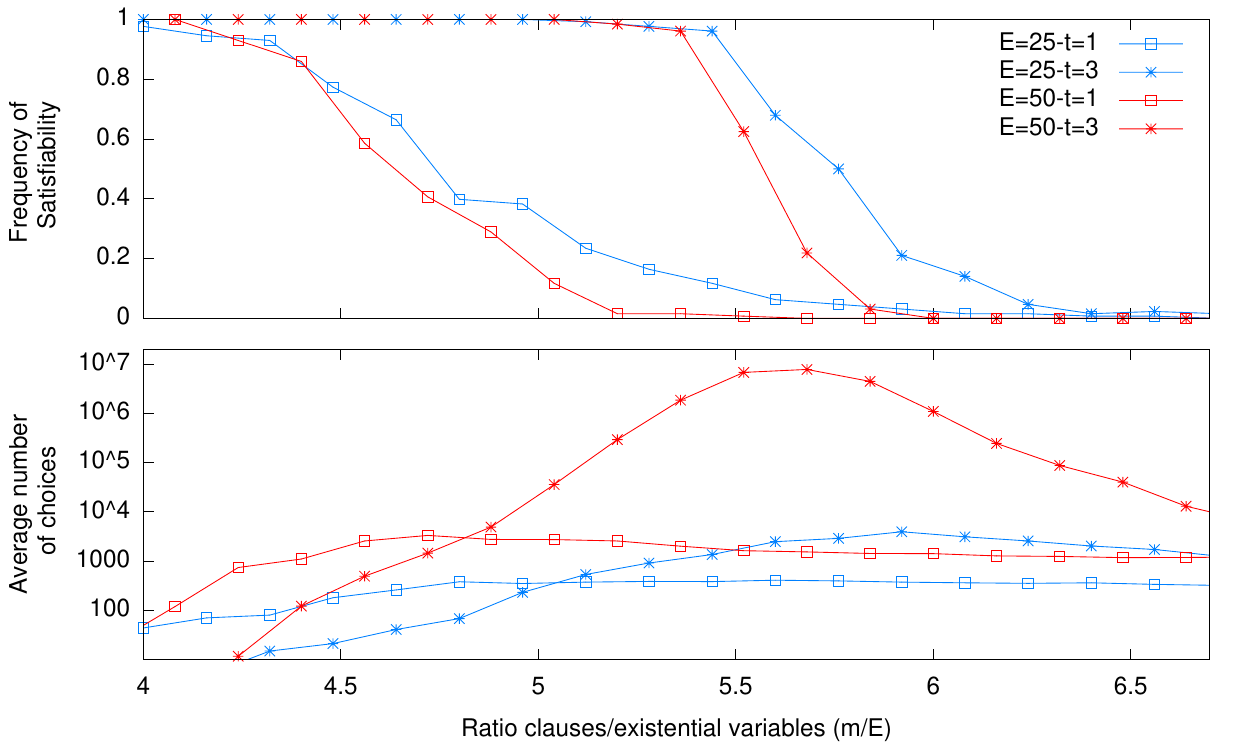} 
    \caption{Combined effect of variables and components (QBF)}\label{fig:bounds:qbf}
	\end{subfigure}
    
%
    \caption{Combined effect of variables and components: phase transition and hardness.}\label{fig:bounds}
\end{figure*}

\subsection{Behavior of Multi-component Model}

To study the satisfiability of multi-component model instances (the 
location of the phase transition), we considered the setting with the number
of variables (propositional atoms) fixed. Figure~\ref{fig:growing:sat} shows
the results for the $t$ component model $t$-$C(3,200,m)$, with 
$t\in\{1,3,5,7,9,11\}$. The $x$-axis gives the ratio of
the numbers of clauses and variables ($m/200$), the $y$-axis shows the 
frequency of SAT. Consistently with our theoretical results, the phase 
transition shifts from left to right, and it sharpens for growing values 
of $t$. The same can be observed in Figure~\ref{fig:growing:qbf}, showing
the frequency of QBFs from 
$t$-$Q(1,3;24,12;m)$ that are true, for
$t\in\{1,3,5,7,9,11\}$. 
The satisfiability plots obtained for logic programs from the corresponding
models $t$-$\dnf_\dlp(1,3;24,12;m)$ are symmetric with respect to the line 
$y=0.5$ and are not reported.

To study the hardness of the multi-component model we computed the average 
solver running times. 
The results (on the same instances as before) for the 
\glucose SAT solver and the \bqcegar QBF solver are in 
Figures~\ref{fig:growing:onesolver:sat}~and~\ref{fig:growing:onesolver:qbf}.
The plots show a strong dependency of the hardness on the 
number of components: the peak of hardness moves right and 
grows visibly with $t$. 
In more detail, the CNF formulas (one component) are solved by \glucose 
in less than 0.42s, whereas instances with 
11 components require about 7 minutes, i.e, they are more than 
3 orders of magnitude harder. Analogous behavior is observed when running
\bqcegar on QBF formulas. Those from the one-component model are solved 
instantaneously (average time $\leq$ 0.01s), those from the 11-component 
model require about one minute. 
The experiments with other solvers gave similar results.

To underline the dependency of the hardness on the number of components, 
for each solver we compute the average time over samples of the same size and plot its maximum (for simplicity \textit{maximum execution time}) for several values of $t$ in
Figures~\ref{fig:growing:sat:maxtime} (SAT)~and~\ref{fig:growing:qbf:maxtime} 
(QBF, programs).
In particular, Figure~\ref{fig:growing:sat:maxtime} reports the results 
obtained by running \glucose and \lingeling, and
Figure~\ref{fig:growing:qbf:maxtime} --- the results obtained by running 
\bqcegar, \aigsolve, \aqua, \rareqs and the results obtained by running \clasp and 
\wasp on the corresponding programs. The picture shows that the peak of 
difficulty grows with the number of components no matter the implementation
or the representation roughly, at a rate that is more than quadratic with 
$t$ ($y$-axis in logarithmic scale).

Next, we discuss the behavior of formulas when both the number of variables 
and the number of components grow.
Figure~\ref{fig:bounds:sat} 
reports on the behavior of CNF formulas with $n\in\{100,200\}$ and $t\in \{1,10\}$. 
Formulas with 100 variables are plotted in red, and those with 200 variables 
in blue. We use squares to identify graphs for formulas with 
one component and stars for graphs concerning formulas with ten
components.  
Figure~\ref{fig:bounds:sat} shows that when the number of variables grows 
the phase transition moves to the left, and the transition becomes 
sharper.
By Theorem \ref{thm:4}, we expect that 
the bounds on (un)satisfiability do not depend on $t$, indeed when the number 
of variables grows the right shift due to an increase in the number of 
components is compensated, and becomes negligible.
Our experiments also confirm that hardness grows with both the number of components and the number of variables. 
This is seen in Figure~\ref{fig:bounds:sat} 
in the bottom, which plots the average number of choices 
taken by \clasp (we consider choices since 
execution times are negligible). 
Note that CNF formulas with 100 variables and 10 components are already harder than formulas with 200 variables and one component. 
Figure~\ref{fig:bounds:qbf} shows the same picture for QBFs 
in $t$-$Q(1,3;50,25;m)$ (plotted in red) and $t$-$Q(1,3;100,50;m)$
(plotted in blue), with $t \in \{1,3\}$. These results were obtained
by running \clasp on the corresponding programs.

\subsection{Combination of Controlled and Multi-component model}
We now present the results obtained by combining the two models presented 
in this paper. We focus on the effect of the combination of models on the 
hardness of formulas.
The results are summarized in Figure~\ref{fig:CIvsCTRL:Histogram} where 
a bar plot depicts the maximum average execution times 
\textcolor{black}{(i.e., the average execution times measured evaluating the hardest instances at the phase transition)}
obtained by running 
ASP and QBF solvers on instances of models $t$-$Q(1,3;A,E;m)$ 
(multi-component with Chen-Interian) and $t$-$Q^{\ctd}(4,A,E)$ 
(multi-component with controlled) varying the number of components 
$t \in\{1,3,5,7,9,11\}$. 
To obtain comparable execution times with both ASP and QBF solvers, 
\clasp and \wasp were run on instances with $E=24$, 
\rareqs and \aqua on instances with $E=18$,
whereas \bqcegar and \aigsolve on instances with $E=12$, and $A\in[2,120]$
and $m\in[2,300]$. Figure~\ref{fig:CIvsCTRL:Histogram} shows histograms for 
each solver. The results obtained for each setting of $t$ in 
$t$-$Q(1,3;A,E;m)$ and $t$-$Q^{\ctd}(4,A,E)$ are reported side by side in 
blue and orange bars, respectively. 
The red horizontal line helps identifying a timeout of 24 hours, 
and a red bar ending with an arrow indicates that some execution required more than 24 hours.
A red exclamation mark identifies abrupt termination of a solver.

\textcolor{black}{We observe that, no matter the solver, the hardest instances of multi-component with controlled are at least one order of magnitude harder than the Chen-Interian-based counterparts for all settings of $t$.} 
Notably, the combination of the two new models allows to generate instance that are ``super-hard'';
indeed instances with one component are solved in less than $0.9$s and it was sufficient to set $t=11$ to obtain instances that are more than six orders of magnitude harder to evaluate (some ``controlled'' instances with $E\geq 18$ could not even be solved in 24 hours).

\begin{figure}[t!]
    \centering  
    \includegraphics[width=\myOneCol]{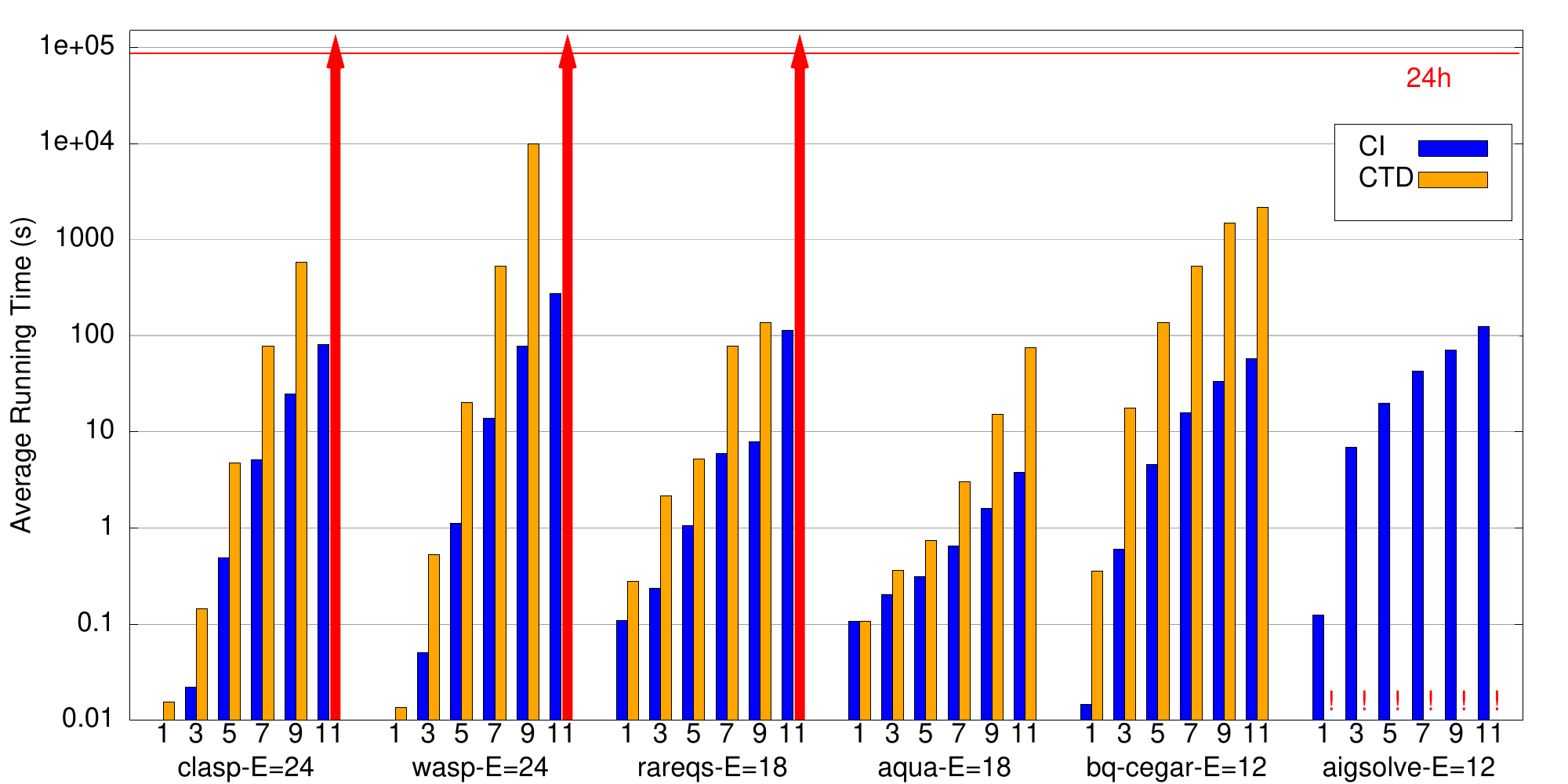}
    \caption{Comparing Chen-Interian and Controlled model: effect of components.}\label{fig:CIvsCTRL:Histogram}
\end{figure}

\subsection{Impact on SAT Solving}
A desirable property of a random model is to generate instances 
that behave similarly to real-world ones~\cite{DBLP:conf/cp/KautzS03,DBLP:conf/ijcai/AnsoteguiBL09}.
%
This similarity has been measured empirically by comparing the performance of solvers for random and industrial instances. 
Following Ans{\'{o}}tegui et al.~\cite{DBLP:conf/ijcai/AnsoteguiBL09},
we measure the ratio of the execution times of solvers.
We compared \kcnfs (a well-known SAT solver specialized in random instances) with \glucose and \lingeling (both specialized in real-world instances) to assess whether our model allows to generate instances that are better solved by solvers for real-world instances.
Figure~\ref{fig:ImpactOnSolving:sat} shows the
results for the model $t$-$C(3,100,m)$, while varying the number of 
components $t\in\{1,2,3,4,5\}$. 
In particular, the $x$-axis gives the ratio of the numbers of clauses and variables ($m/100$), 
and  the $y$-axis shows \glucose versus \kcnfs (in Figure~\ref{fig:ImpactOnSolving:sat:glu}) 
and \lingeling versus \kcnfs (in Figure~\ref{fig:ImpactOnSolving:sat:lin}).

We observe that, \kcnfs is faster (ratios $>1$) than both \glucose and \lingeling when $t=1$, i.e., when our model coincides with the classical one for random formulas. 
Once we increase the number of components the result is reverted, \glucose and \lingeling are faster than \kcnfs (ratios $<1$), and the difference grows significantly with $t$.
This is independent of the clauses/variables ratio.

The difference between random and real-world instances is 
often attributed to the presence of some \textit{hidden structure} in 
the latter~\cite{DBLP:conf/aaai/AnsoteguiBLM08}. We observed that
multi-component models yield instances that are solved faster by 
solvers designed for real-world instances. We conjecture this is due 
to the component structure introduced by the model. This structure can 
be controlled by varying the number of components, yielding instances 
of varying hardness.

\begin{figure*}[t!]
    \centering
    \captionsetup[sub]{font=normal,labelfont={sf,sf}}
    	\vspace*{1.2cm}
    \begin{subfigure}[t]{\columnwidth}\centering
        \includegraphics[width=\myOneCol]{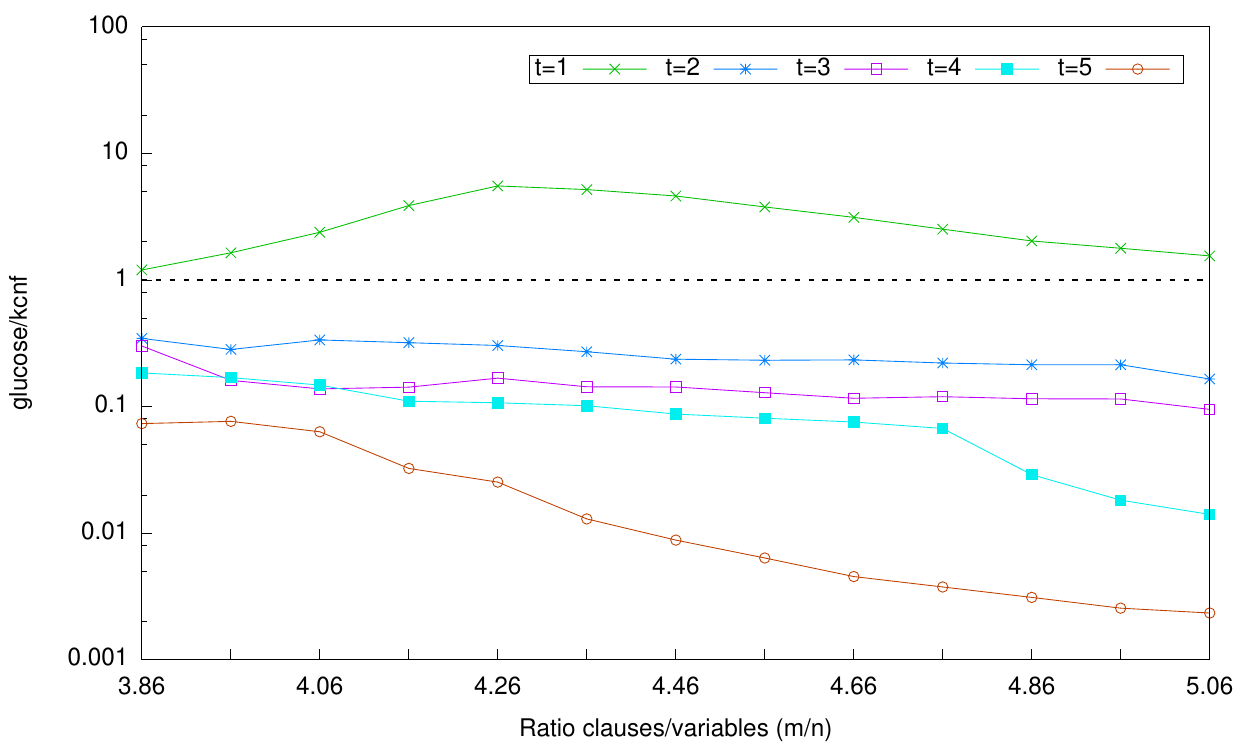}
    \caption{CPU time ratio glucose/kcnfs}\label{fig:ImpactOnSolving:sat:glu}
	\end{subfigure}    
    	\vspace*{0.8cm}
    	
    \begin{subfigure}[t]{\columnwidth}\centering
        \includegraphics[width=\myOneCol]{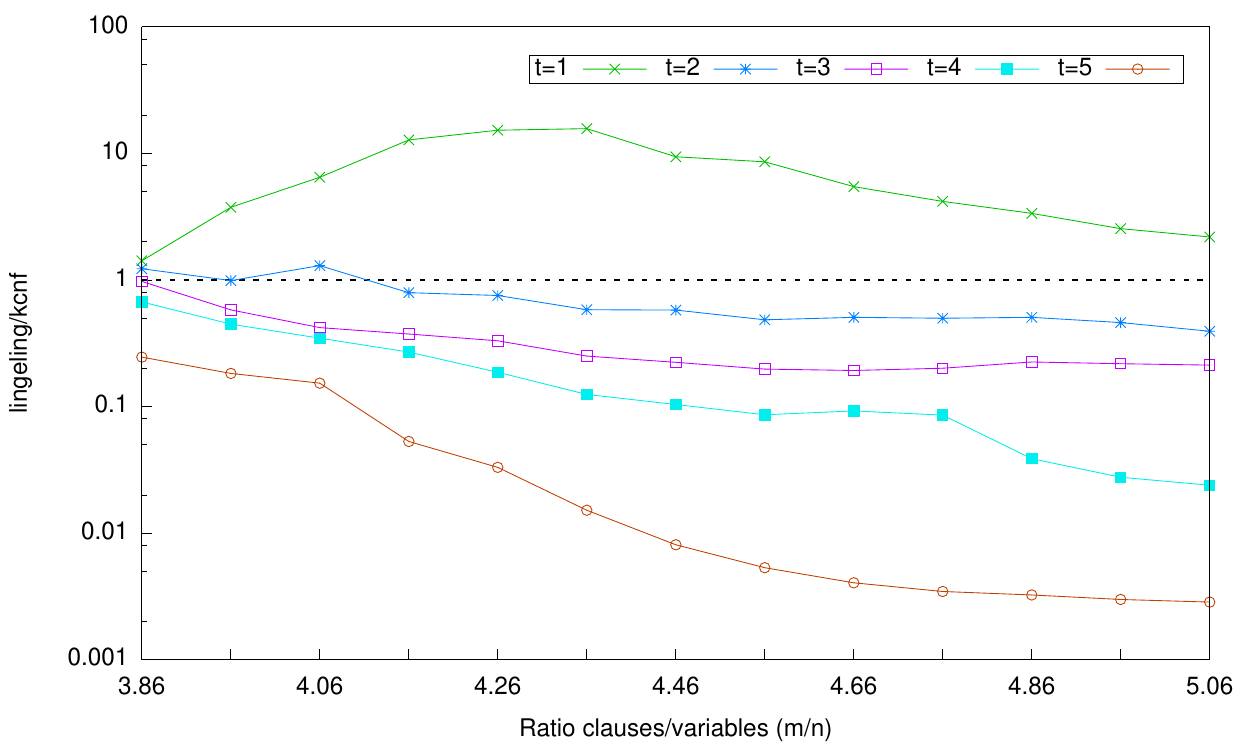}
    \caption{CPU time ratio lingeling/kcnfs}\label{fig:ImpactOnSolving:sat:lin}
	\end{subfigure}    

    \caption{Impact on SAT solving.}\label{fig:ImpactOnSolving:sat} 
\end{figure*}

\subsection{Impact on QBF and ASP Solving}
An analysis distinguishing the behavior of random and industrial instances 
is not possible for ASP and QBF solvers. 
Indeed, no QBF/ASP solvers have ever been designated (or known) as 
\textit{specialized to random instances} in ASP and QBF Evaluations so far 
(cf. \cite{CalimeriGMR16,NarizzanoPT06} and \url{http://www.qbflib.org}).
Nonetheless our models has other interesting implications for QBF and ASP solvers.

\paragraph{Impact on QBF Solving}
To assess the validity of our multicomponent Chen-Interian model for QBF, 
we submitted several instances to the QBF Evaluation 2016.
All our instances (with $n=100$ only, and $t \leq 6$) were classified as \textit{hard} by the organizers, and helped identify a bug in one of the participating solvers, 
demonstrating the efficacy of our model in performance analysis 
and in correctness testing.

\paragraph{Impact on ASP Solving}
For ASP solvers, 
Figure~\ref{fig:ImpactOnSolving:asp:ci} outlines the impact of our model on answer set search for programs corresponding to QBF formulas $t$-$Q(1,3,24,12,m)$ with $t \in \{ 1,3,5,7,$ $9, 11\}$. 
ASP solvers evaluate disjunctive programs by first computing a candidate 
model, and then checking its stability (the latter task is co-NP complete). 
Thus, we plot $(i)$ the ratio between the number of choices made during the search phase and the number of stable model checks performed by \wasp and \clasp, and $(ii)$ the ratio between the time spent in stable model checking and the total execution time for the solver \wasp (results for \clasp are analogous) 
both for growing $t$. 
The ratio between the numbers of choices and model checks 
decreases when the number of components grows, following a similar behavior 
for both solvers. This is a machine-independent measure 
of the impact of the two activities, and we observe that 
the role of the model checker grows with $t$.
Specifically, the impact of the model checking on the total solving time 
grows from about about 3\% ($t=1$) to 88\% ($t=11$). 
Analogous considerations are supported by Figure~\ref{fig:ImpactOnSolving:asp:ctrl}, which
outlines the impact of the combination of controlled and multi-component model on answer set search for programs corresponding to QBF formulas $t$-$Q^{ctd}(4,32,16,m)$ with $t \in \{ 1, 3, 5, 7, 9, 11\}$. Also in this case, and increase of $t$ causes both $(i)$ a decrease of the ratio between the number of choices made during the search phase and the number of stable model checks, and $(ii)$ an increase of the the time spent in stable model checking and the total execution time for the solver
\wasp (results for \clasp are analogous).
Specifically, the impact of the model checking on the total solving time 
grows from about about 0.05\% ($t=1$) to 51\% ($t=11$).

\begin{figure*}[t!]
    \centering
    \captionsetup[sub]{font=normal,labelfont={sf,sf}} 
    \vspace*{1.2cm}
    \begin{subfigure}[t]{\columnwidth}\centering
    	\includegraphics[width=\myOneCol]{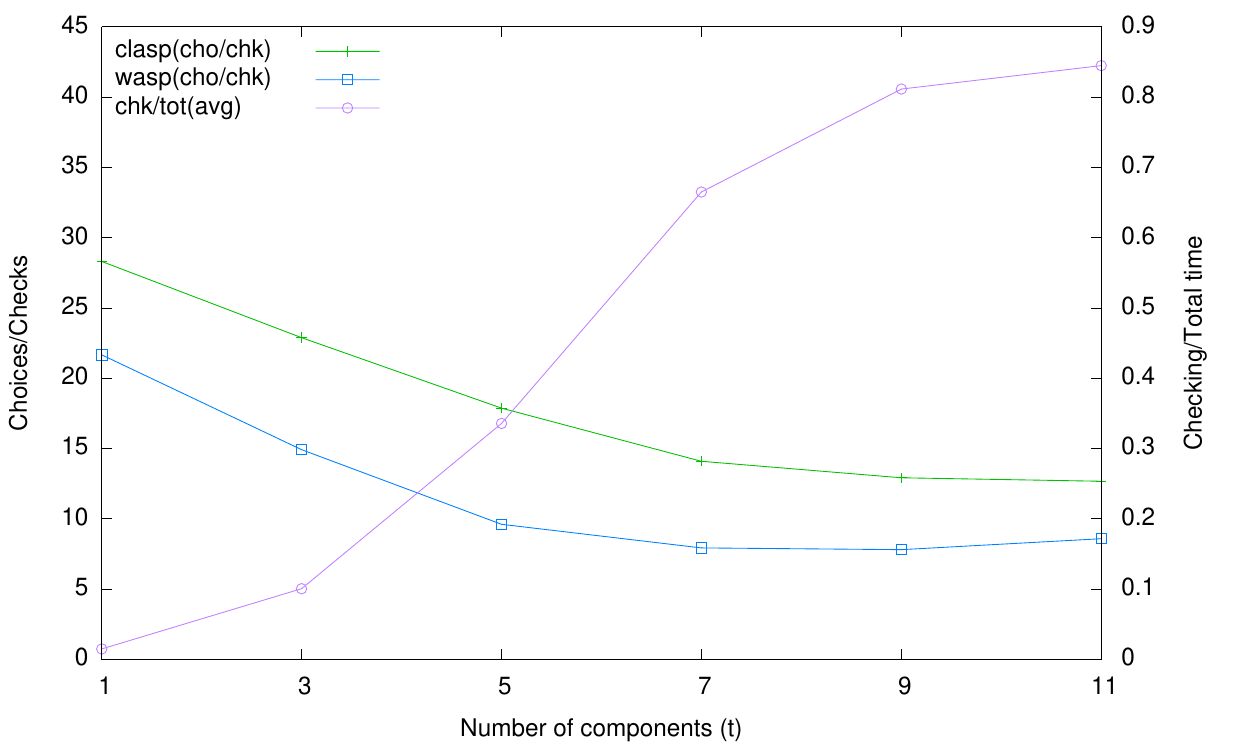} 
    \caption{Answer set computation: Multi-component}\label{fig:ImpactOnSolving:asp:ci}
	\end{subfigure}    
		\vspace*{0.8cm}
		
    \begin{subfigure}[t]{\columnwidth}\centering
	\includegraphics[width=\myOneCol]{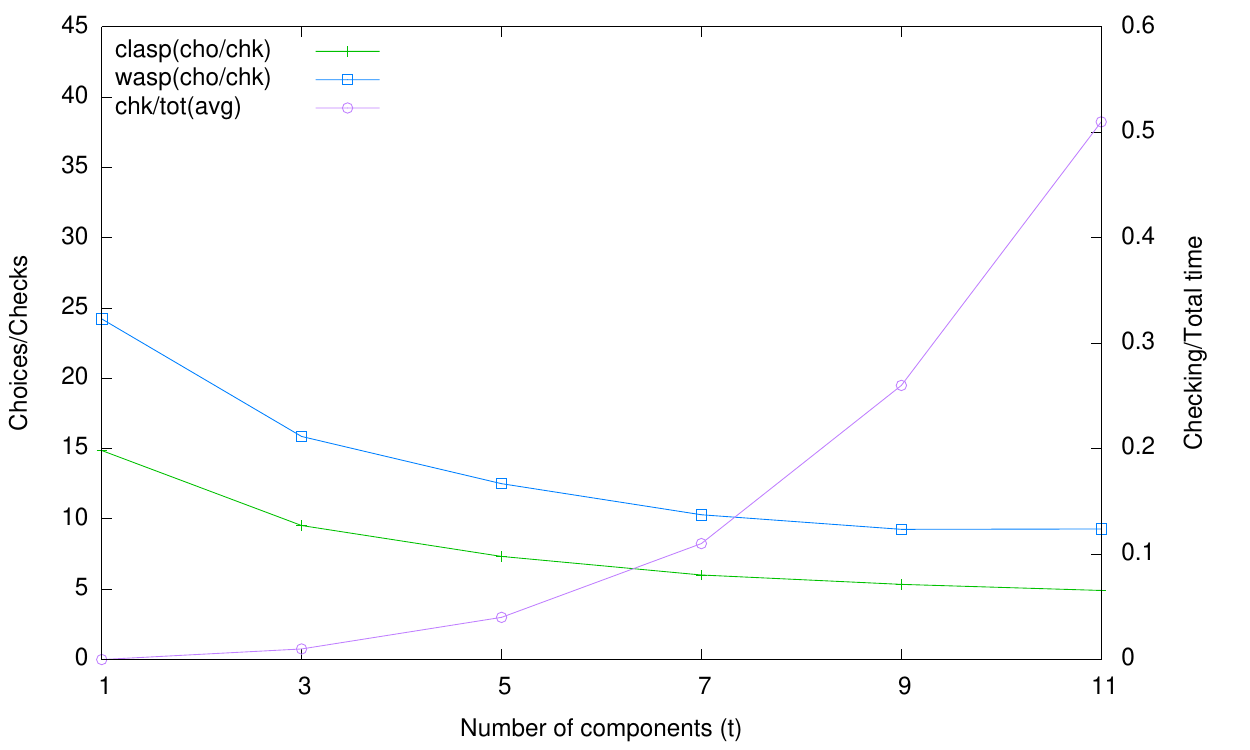} 
    \caption{Answer set computation: Multi-component Controlled Model}\label{fig:ImpactOnSolving:asp:ctrl}
	\end{subfigure}       
%
    \caption{Impact on ASP solving: answer set computation.}\label{fig:ImpactOnSolving:asp} 
\end{figure*}

\begin{figure*}[t!]
    \centering
    \captionsetup[sub]{font=normal,labelfont={sf,sf}} 
    	\vspace*{1.2cm}
    \begin{subfigure}[t]{\columnwidth}\centering
	\includegraphics[width=\myOneCol]{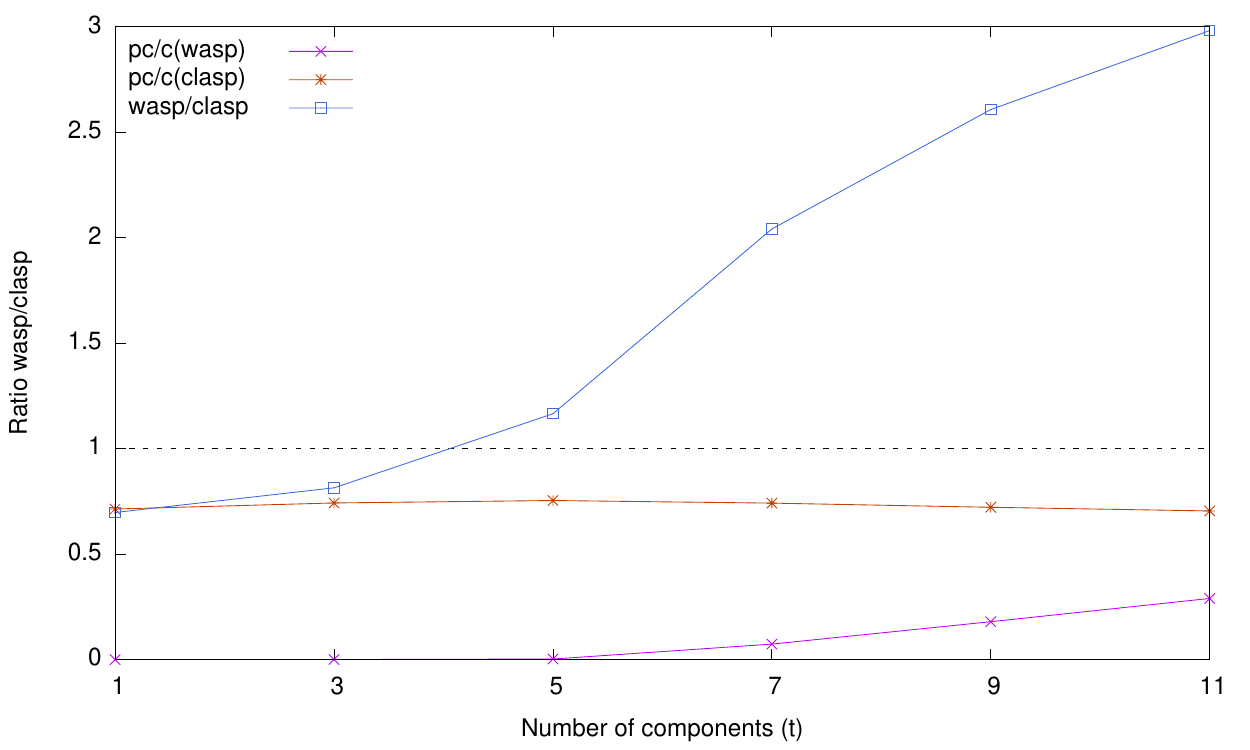} 
    \caption{Partial checking: Multi-component}\label{fig:ImpactOnSolving:pc:ci}
	\end{subfigure}    
	\vspace*{0.8cm}
	
    \begin{subfigure}[t]{\columnwidth}\centering
	\includegraphics[width=\myOneCol]{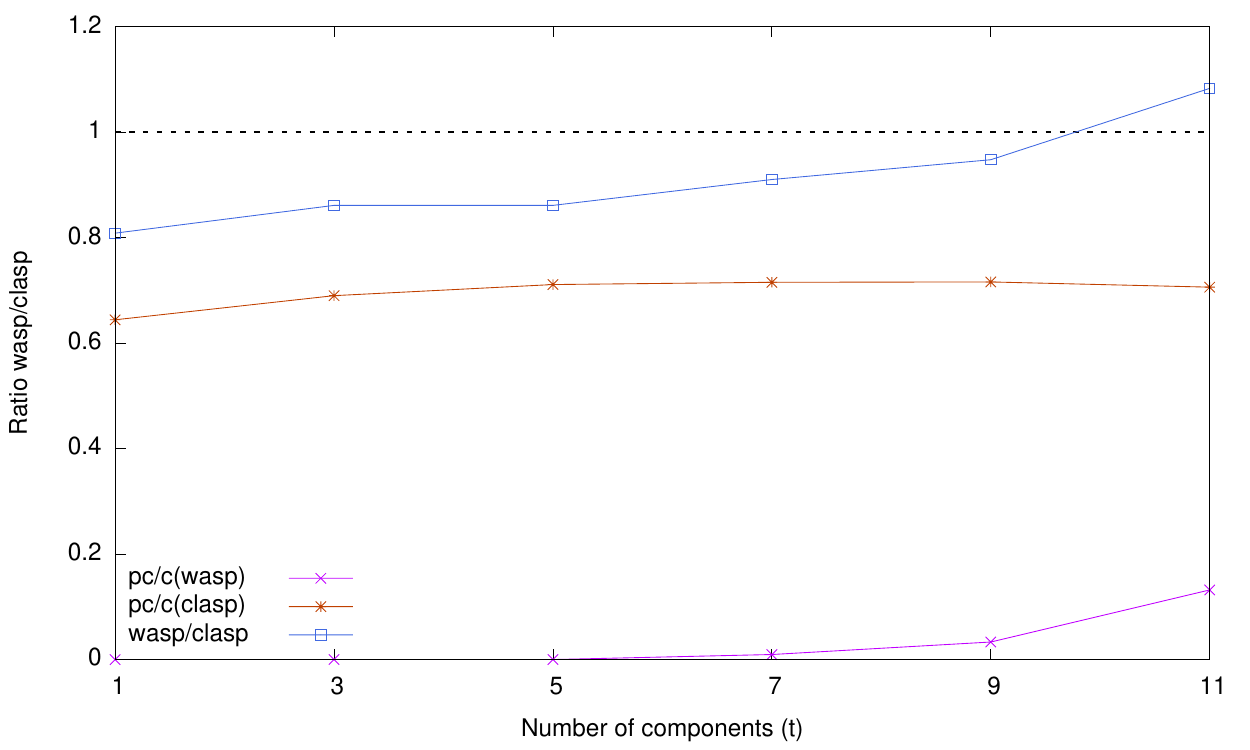} 
    \caption{Partial checking: Multi-component Controlled Model}\label{fig:ImpactOnSolving:pc:ctrl}
	\end{subfigure}    

%
    \caption{Impact on ASP solving: impact of partial checking.}\label{fig:ImpactOnSolving:asp:pc} 
\end{figure*}

It is known that on usual benchmarks ASP solvers spend more time in 
the model search phase than in the final model checking phase~\cite{DBLP:conf/lpnmr/Pfeifer04} 
(this also happens on benchmarks we generated for $t=1$). However, our multi-component models 
allow us to generate in a controllable way instances that put emphasis 
on the model checking phase.

Finally, we report some other observations that point to a potential 
impact of our models in detecting areas of improvement for solvers.
Let us recall that the two ASP solvers we studied, \clasp and \wasp,
employ different strategies for stable model checking. \clasp searches 
for unfounded sets~\cite{DBLP:conf/ijcai/GebserKS13}, while \wasp searches 
for a minimal model of the program reduct~\cite{AlvianoDLR15}. Both solvers 
are able to check partial interpretations, but they employ different 
heuristics for enabling this search space pruning technique.
Figure~\ref{fig:ImpactOnSolving:asp:pc} compares \wasp and \clasp by 
plotting the ratio between the time required by the two solvers for finding 
an answer set (labeled \wasp/\clasp) and the ratio between the number of 
partial 
and total checks performed by \wasp (labeled pc/c(\wasp)) and \clasp 
(labeled pc/c(\clasp)) for the two multi-component models we studied.
The results on the multi-component Chen-Interian model are in  
Figure~\ref{fig:ImpactOnSolving:pc:ci}. The results on the multi-component 
controlled model are in Figure~\ref{fig:ImpactOnSolving:pc:ctrl}. One can
see that \wasp is faster than \clasp when the number of components 
is small. When the number of components grows \clasp becomes faster 
and takes over. Interestingly, the deterioration in the 
performance of \wasp corresponds to the point in which the ratio pc/c 
starts increasing. In contrast, \clasp maintains consistently the ratio
of about 70\% of the numbers of partial and total checks, and this seems 
to pay off for larger values of $t$. The results suggest that partial 
checking in \wasp was implemented in a less efficient way then in \clasp, 
and it hinders \wasp when the number of components grows. It seems also 
that for easier instances better performance could be obtained by disabling 
or reducing the number of partial checks as they do not 
seem to be essential for the performance. These observation suggest that 
there is space for solver developers to devise smarter heuristics for 
improving the usage of partial checking.

\section{Conclusions}
In this paper we proposed the controlled and multi-component models for random 
propositional formulas, and disjunctive logic programs. 
The models extend the well-known fixed clause length model for k-SAT, 
and the Chen-Interian model for QBF. 

We provide theoretical bounds that predict the location of 
the region where the phase-transition occurs, and we present the results 
of an experimental analysis that confirms our theoretical findings in 
practice. Our experiments also show that the hardest instances are located 
in the phase transition region. Moreover, in the multi-component model the 
hardness of formulas depends significantly on the number of components.

Comparing models, we observed that the controlled model allows one to 
generate random instances that are much harder than those obtained with 
the Chen-Interian model with the same number of existential variables. 
Further, multi-component model allows one to generate random instances with 
few components that are several orders of magnitude harder than those 
generated with the same number of variables from the underlying 
``single-component'' model. Finally, a combination of the two new models 
results in the generation of programs and formulas that are ``super-hard'' 
to evaluate.

Our experiments with different solvers and encodings gave consistent results.
This supports our claim that the phenomena we observed are inherent properties 
of the models rather than an artifact of the solver used.

Despite their simple structure the models have theoretical and empirical 
properties that make them important for further advancement of the SAT, QBF 
and ASP solvers.

First, the hardness of formulas/programs can be controlled and, unlike 
in the earlier models, not only in terms of the ratio of clauses to variables. 
Our experiments showed that the hardness \emph{strongly} depends on the number 
of components. Thus, it can also be controlled by varying that parameter,
and even a small number of components can lead to extremely hard instances.
Further, in our experiments (as well as in the QBF Competition 2016) instances 
generated according to our models helped identifying bugs in existing solvers.  
Moreover, the multi-component model generates formulas that in at least one
aspect are similar to instances arising in practice: they are solved better 
by SAT solvers specialized in industrial benchmarks than by SAT solvers 
specialized in random ones. 
This makes them useful for development and testing of solvers intended
for practical applications. Finally, our models of random disjunctive
programs are the first such models for that class of logic objects. This
and the fact that it allows us to control the role of the stable model
checking phase point to its potential for the development of ASP solvers.

Our work raises an interesting open question. The controlled model we proposed
and studied stipulates that clauses in the matrix of QBFs contain exactly one
universal variable. It is possible to lift this requirement. We discuss some 
natural extensions in \ref{app:generalization}. It turns out that when the 
number of universal variables per clause is greater than one, the generalized 
model generates instances that exhibit a qualitatively different behavior. 
Arguably, they are easier then formulas from the corresponding Chen-Interian 
model.

However, a comparison to the Chen-Interian model is not clear cut, a problem 
we already noted for the one universal variable case. In particular, we chose
to compare for hardness formulas from the two models by fixing in each model 
the number of existential variables to the same value. Under this constraint,
the hardest formulas in the basic controlled model contain more universal 
variables than it is the case for the hardest formulas from the Chen-Interian 
model. However, for the generalized controlled model and its \emph{smooth} 
version, both discussed in \ref{app:generalization}, this relationship 
reverses. The hardest formulas from the (smooth) generalized controlled model 
have \emph{many} fewer universal variables than the hardest ones from the 
Chen-Interian model. Developing alternative perspectives on formulas from 
the two models might provide a better understanding of the relative hardness.
This is an important avenue to explore and it requires further studies.

\section*{Acknowledgments}
\thanks{
The work of the first two authors has been partially supported by 
the Italian Ministry for Economic Development (MISE) under the project
``PIUCultura -- Paradigmi Innovativi per l'Utilizzo della Cultura'' 
(n. F/020016/01-02/X27), and under project ``Smarter Solutions in the Big 
Data World (S2BDW)'' (n. F/050389/01-03/X32) funded within the call 
``HORIZON2020'' PON I\&C 2014-2020. The work of the third author has been
partially supported by the the NSF grant IIS-1707371.
}

\bibliographystyle{abbrv}
\bibliography{bibtex}

\label{lastpage}

\newpage

\section*{Appendix}

\subsection*{Generalized Controlled Model}\label{app:generalization}

The controlled model $Q^\ctd(k,A,E)$
introduced in Section~\ref{sec:controlled}, stipulates 
that every clause in the matrix of a QBF from the model contains exactly one
universal variable, and that each universal variable occurs in exactly two
clauses, in one of them as a positive literal (not negated) and in the other
one as the negative literal (negated). It follows that $Q^\ctd(k,A,E) 
\subseteq Q(1,k-1;A,E;2A)$, that is, the controlled model $Q^\ctd(k,A,E)$
is a restriciton of the Chen-Interian model $Q(1,k-1;A,E;2A)$. Moreover, 
the key property of the controlled model $Q^\ctd(k,A,E)$ is that, for every
truth assignment to the universal variables in $X$, once we simplify the 
matrix accordingly we are left with \emph{exactly} $A$ $(k-1)$-literal clauses
over $E$ variables, whereas in the case of the Chen-Interian model 
$Q(1,k-1;A,E;2A)$, similar simplifications leave us with 
with varying number of $(k-1)$-clauses, with the \emph{average} number being 
$A$.

We now generalize the model to allow clauses with exactly $h$ occurrences of 
universal variables, where $h$ is a fixed integer satidfying $1\leq h\leq k$.
More precisely, the model consists of QBFs $\forall X \exists Y\ F$, where
$F$ consists of $\binom{A}{h} 2^h$ $k$-literal clauses, each clause consists 
of $h$ literals over $X$ and $k-h$ literals over $Y$ (with no repetitions of 
variables), and where for every consistent set of $h$ literals over $X$ there 
is a single clause in the formula that contains them. A QBF in this model 
(to be precise, its matrix) is obtained by generating $\binom{A}{h} 2^h$ 
$h$-literal clauses over $X$ and extending each of them by a randomly 
generated consistent $(k-h)$-element set of literals over $Y$. We denote the 
set of QBFs obtained in this way by $Q^\gctd(h,k-h,A,E)$ and call it the 
\textit{generalized controlled model}. 

Clearly, $Q^\ctd(k,A,E)=Q^\gctd(1,k-1,A,E)$. Thus, the controlled model
we discussed in the paper is a special case of the model described here.
We also note that $Q^\gctd(h,k-h,A,E) \subseteq Q(h,k-h;A,E;\binom{A}{h} 2^h)$. 
Thus, the generalized controlled model is a restriction of the appropriate
Chen-Interian model --- while random with respect to variables in $Y$ 
(existentially quantified variables), the way variables in $X$ (universally
quantified variables) are treated is fully deterministic. In particular, for
every truth assignment on $X$, once we simplify the matrix accordingly, 
we are left with \emph{exactly} $\binom{A}{h}$ $(k-h)$-literal clauses 
over $E$ variables in $Y$, while in the case of the Chen-Interian model 
$Q(h,k-h;A,E;\binom{A}{h} 2^h)$, similar simplifications leave us with
$(k-h)$-CNF formulas with varying number of clauses, with the 
\emph{average} number being $\binom{A}{h}$.

\begin{example}
Consider a set $X=\{x_1,x_2,x_3\}$ of $A=3$ universal variables and a set 
$Y=\{y_1,y_2,y_3,y_4\}$ of $E=4$ existential variables. We are interested 
in ``generalized controlled formulas'' having $k=5$ literals with $h=2$ of
them over $X$. That is, we are interested in the model $Q^\gctd(2,3,3,4)$.
According to the definition, we have to build $\binom{3}{2}2^2=12$ clauses 
of length $5$. For an appropriate enumeration $C_i$, $i=1,\ldots,12$, these 
clauses will satisfy:
\begin{center}
$\begin{array}{rclrclrcl}
\{x_1,x_2\}&\subseteq& C_1\hspace{1cm} &
\{x_1,x_3\}&\subseteq& C_5\hspace{1cm} &
\{x_2,x_3\}&\subseteq& C_9 \\
\{\neg x_1,x_2\}&\subseteq& C_2\hspace{1cm} &
\{\neg x_1,x_3\}&\subseteq& C_6\hspace{1cm} &
\{\neg x_2,x_3\}&\subseteq& C_{10} \\
\{x_1,\neg x_2\}&\subseteq& C_3\hspace{1cm} &
\{x_1,\neg x_3\}&\subseteq& C_7\hspace{1cm} &
\{x_2,\neg x_3\}&\subseteq& C_{11} \\
\{\neg x_1,\neg x_2\}&\subseteq& C_4\hspace{1cm} &
\{\neg x_1,\neg x_3\}&\subseteq& C_8\hspace{1cm} &
\{\neg x_2,\neg x_3\}&\subseteq& C_{12}
\end{array}$
\end{center}
Let us choose any truth assignment $\sigma$ on $X$, for instance, 
$\sigma(x_1)=x_1$, $\sigma(x_2)=\neg x_2$, and $\sigma(x_3)=\neg x_3$.
Once we simplify the clauses with respect to this assignment, exactly
$\binom{3}{2}=3$ clauses 
$C_2\setminus\{\neg x_1,x_2\}$, 
$C_6\setminus\{\neg x_1,x_3\}$ and 
$C_{9} \setminus\{x_2,x_3\}$ 
remain 
(all other
clauses after the simplifications become tutologies and can be dropped). 
\end{example}

Let $q^\gctd(h,k-h,A,E)$ denote the probability that a random formula in 
$Q^{\gctd}(h,k-h,A,E)$ is true. 
We define $\mu^\gctd_l(h,k-h)$
to be the supremum over all positive real numbers $\rho$ such that
$$\lim_{E\rra\infty} q^\gctd(h,k-h,\lfloor \rho E^{1/h}\rfloor,E) =1,$$
and $\mu^\gctd_u(h,k-h)$ to be the infimum over all positive real numbers $\rho$ such that 
$$\lim_{E\rra\infty} q^\gctd(h,k-h,\lfloor \rho E^{1/h}\rfloor, E) =0.$$

We will now derive bounds on $\mu^\gctd_l(h,k-h)$ and $\mu^\gctd_u(h,k-h)$ by 
exploiting results on random $(k-h)$-CNF formulas. The proof is an adaptation 
of the proof of Theorem \ref{th:controlledThreshold}.

\begin{thm}\label{th:genControlledThreshold} 
For every integers $k$ and $h$ such that $k\geq 2$ and $1\leq h<k$, 
$\mu^\gctd_l(h,k-h)$ and $\mu^\gctd_u(h,k-h)$ are well defined.
\end{thm}
\begin{proof}
Let $\Phi\in Q^{\gctd}(h,k-h,A,E)$, $X=\{x_1,...,x_A\}$, and 
$Y=\{y_1,...,y_E\}$. By the definition, $\Phi=\forall X \exists Y F$, where 
$F=C_1 \wedge\ldots\wedge C_N$ is a $k$-CNF formula of $N=\binom{A}{h}2^h$ 
clauses $C_i= l_{i,1}\vee\ldots\vee l_{i,k}$ such that $l_{i,1},\ldots,
l_{i,h}$ are literals over $X$ and $l_{i,h+1},\ldots,l_{i,k}$ are literals 
over $Y$. We define $C_i^Y=l_{i,h+1}\vee\ldots\vee l_{i,k}$ and $F^Y = 
C_1^Y\land\ldots\land C_N^Y$. Moreover, for every 
interpretation $I$ of $X$ we define $F|_I=\bigwedge\lbrace C_i^Y \mid 
C_i\in F \mbox{ and } I\not\models l_{i,1}\lor\ldots\lor l_{i,h} \}$.

Let us assume that $\Phi$ is selected from $Q^{\gctd}(h,k-h,A,E)$ uniformly 
at random. By the definition of the model $Q^\gctd(h,k-h,A,E)$, $F^Y$ can 
be regarded as selected from $C(k-h,N,E)$ uniformly at random and, for each 
truth assignment $I$ of $X$, $F|_I$ can be regarded as selected uniformly at 
random from $C(k-h,M,E)$, where $M=\binom{A}{h}$. 

To show that $\mu^\gctd_u(h,k-h)$ are  well defined, it is enough to show that
there are $r$ and $s$ such that
\[
\lim_{E\rra\infty} q^\gctd(h,k-h,\lfloor r E^{1/h}\rfloor, E) =0
\quad\mbox{and}\quad
\lim_{E\rra\infty} q^\gctd(h,k-h,\lfloor s E^{1/h}\rfloor,E) =1.
\]
The proof relies on an obvious property that for every fixed positive 
integer $h$, there are positive constants $\alpha_h$ and $\beta_h$ such that 
for every sufficiently large positive integer $A$,
\[
\beta_h A^h \geq \binom{A}{h} \geq \alpha_h A^h.
\]
 
To prove the existence of $r$, let us fix any real $r$ such that 
$\alpha_h (r/2)^h > \rho_u(k-h)$,
and let $A=\lfloor r E^{1/h}\rfloor$. Next, let $\Phi =\forall X \exists 
Y F$ be a QBF selected uniformly at random from $Q^{\gctd}(h,k-h,A,E)$ and 
$I$ be a truth assingment on $X$. Clearly, if $F|_I$ is unsatisfiable, then 
$\Phi$ is false. 

For all sufficiently large $E$, we have $A \geq (r/2) E^{1/h}$. 
Consequently, $A^h \geq (r/2)^h E$ and
\[
\binom{A}{h}\geq \alpha_h (r/2)^h E.
\]
Since $\alpha_h(r/2)^h > \rho_u(k-h)$, it follows that the probability 
that $F|_I$ is unsatisfiable tends to 1 with $E$. Thus, the probability that 
$\Phi$ is false tends to 0 with $E$, too. In other words,
\[
\lim_{E\rra\infty} q^\gctd(h,k-h,\lfloor r E^{1/h}\rfloor, E) =0.
\]

To prove the  existence of $s$ we proceed similarly. Let $s$ be any positive 
real such that $2^h\beta_h s^h\leq \rho_l(k-h)$ and let $A =\lfloor s 
E^{1/h}\rfloor$. Further, as before, let $\Phi =\forall X \exists Y F$ be 
a QBF selected uniformly at random from $Q^{\gctd}(h,k-h,A,E)$. 

Clearly, if the formula $F^Y$ is satisfiable, then for every interpretation 
$I$ of $X$, the formula $F|_I$ is satisfiable or, equivalently, $\Phi$ is 
true. In our case, we have that $A\leq s E^{1/h}$. Thus, $A^h \leq s^h E$.
It follows that $2^h\binom{A}{h} \leq 2^h \beta_h s^h E$.
Thus, $2^h\binom{A}{h}/E \leq 2^h \beta_h s^h < \rho_l(k-h)$. It follows
that the probability that $F^Y$ is satisfiable tends to 1 with $E$ and so, 
the probability that $\Phi$ is true tends to $1$ with $E$. In other words,
\[
\lim_{E\rra\infty} q^\gctd(h,k-h,\lfloor s E^{1/h}\rfloor,E) =1.
\]

\vspace*{-0.2in}
\end{proof}

\renewcommand{\myOneCol}{0.93\columnwidth}
\begin{figure*}[t!] 
    \centering
    \includegraphics[width=\myOneCol]{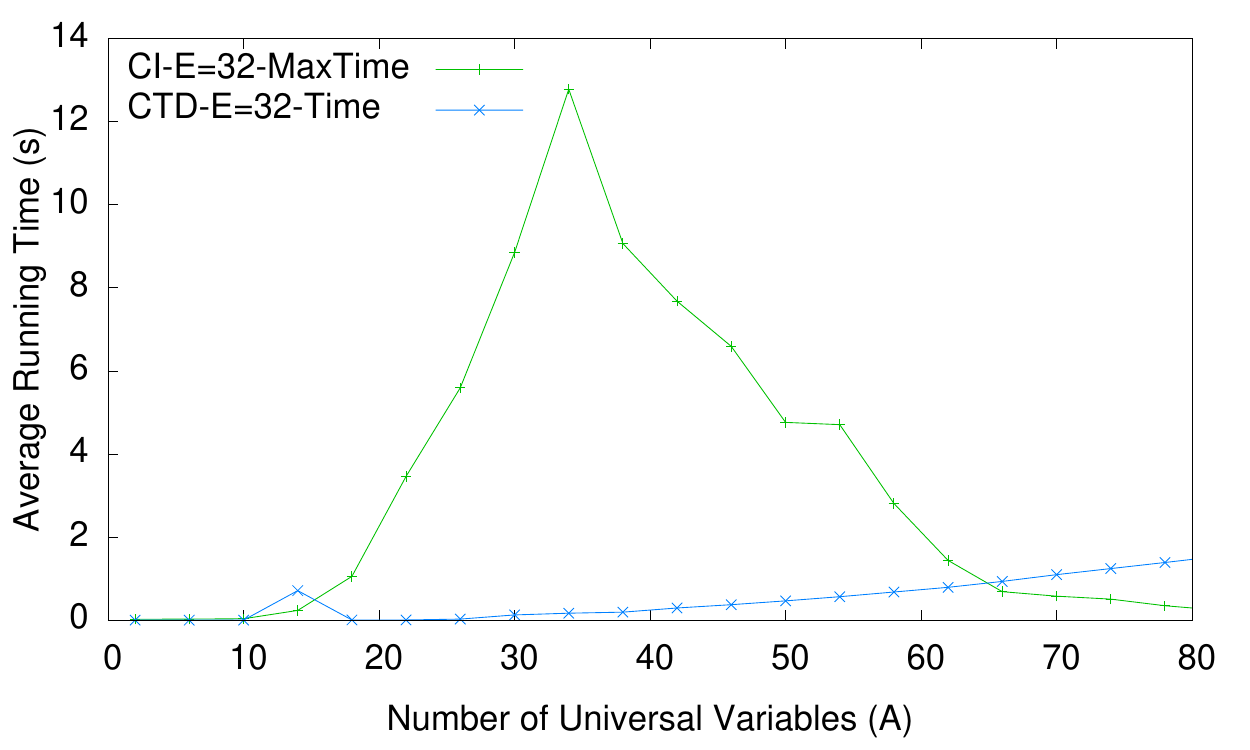}
    \caption{Comparing Chen-Interian and Generalized Controlled model: hardness comparison.}\label{fig:CTRLGeneral}
\end{figure*}

\paragraph{Empirical Behavior.}
We now discuss properties of the generalized controlled model presented above.
In particular we compare the generalized controlled and the Chen-Interian 
models with respect to the hardness of formulas having the same number of 
variables, and comment on one possible extension of the generalized model. 

We consider formulas from the Chen-Interian model $Q(a,e;A,E;m)$, where we set $a=2$, $e=3$, and $E=32$, 
and vary the number $A$ of universal variables over the range $[2..80]$
and the number $m$ of clauses over the range $[10..1200]$.
As we did in Section~\ref{cvci}, we compare the hardness properties of the generalized controlled and the Chen-Interian models with the same number of existential variables. 
These results are shown in Figure~\ref{fig:CTRLGeneral}.
For the controlled model, for each value of $A$, the value on the corresponding hardness graph (the blue line) is obtained by averaging the solve times on formulas generated from the model $Q^{\gctd}(2,3,A,32)$. 
The matrices of these formulas are 5-CNF formulas over $A+32$ variables and with $4\binom{A}{2}$ clauses. The 
corresponding point on the hardness graph for the Chen-Interian model is 
obtained by averaging the solve times on formulas generated from the model 
$Q(2,3;A,32;max)$, where for each $A$ and $E=32$, $max$ is selected to 
maximize the solve times (in particular, it falls in the phase transition 
region for the combination of the values $A$ and $E=32$). The matrices of
these formulas are 5-CNF formulas over $A+32$ variables and $m$ clauses.
The results show that the peak hardness regions for the two models are not 
aligned. Comparing this results with the one in 
Figure~\ref{fig:CIvsCTRL:Comparison} we note that the generalized controlled 
model instances are much easier to solve than Chen-Interian ones, almost in 
every setting. The peak hardness from the generalized controlled model 
instances happens before the maximum hardness the peak hardness region for 
the Chen-Interian model. This is the opposite of what happens for (basic) 
controlled model instances as shown in Figure~\ref{fig:CIvsCTRL:Comparison}.

One possible weakness of the generalized controlled model is that the number 
of clauses, $m=4\binom{A}{2}$, grows quadratically with the number $A$ of 
universal variables. Informally, this growth creates ``long jumps'' in terms 
of the number of clauses in a formula as we increment $A$ and so, also the 
corresponding jumps in the ratio of the number of clauses to the number of 
existential variables. That may cause the model to miss the ``sweet spot'' 
of maximum hardness. For example, already in our experiment with $h=2$, 
formulas with $A=14$ feature 364 clauses, and formulas with $A=15$ feature 
420 clauses. We established experimentally that formulas with $A=14$ are
satisfied with the frequency $~0.1$, whereas the frequency of a satisfiable
instance for $A=15$ is 1.

In order to verify whether the ``jumps'' contribute to the generation of 
easier formulas, we further extended the generalized controlled model to
fill the gaps. Specifically, the \textit{smooth} generalized controlled model,
denoted by $Q^\sgctd(h,k-h,E;m)$, where we specify the number of existential
variables and the number of clauses in the matrix, and where the number of
universal variables is determined by the constraint $2^{h} \binom{A-1}{h}+1 
\leq m \leq 2^{h} \binom{A}{h}$. In particular, if $m=2^h\binom{A}{h}$,
$Q^\sgctd(h,k-h,E;m)$ is defined to coincide with the generalized controlled 
model $Q^\gctd(h,k-h,A,E)$. Formulas for $m$ satisfying $2^{h} \binom{A-1}{h}+1
\leq m<2^h\binom{A}{h}$ are obtained by generating an instance of 
$Q^\gctd(h,k-h,A,E)$ and randomly choosing $m$ among its clauses. 

A phase transition result holds also for the smooth generalized controlled
model. Let $q^\sgctd(h,k-h,E;m)$ denote the probability that a random formula 
in $Q^{\sgctd}(h,k-h,E;m)$ is true. We define $\mu^\sgctd_l(h,k-h)$
to be the supremum over all positive real numbers $\rho$ such that
$$\lim_{E\rra\infty} q^\sgctd(h,k-h,E,\lfloor \rho E\rfloor) =1,$$
and $\mu^\sgctd_u(h,k-h)$ to be the infimum over all positive real 
numbers $\rho$ such that $$\lim_{E\rra\infty} q^\sgctd(h,k-h,E,\lfloor
\rho m\rfloor) =0.$$ Theorem \ref{th:genControlledThreshold} implies
the following result.

\begin{corollary}
For every integers $k$ and $h$ such that $k\geq 2$ and $1\leq h<k$,
$\mu^\sgctd_l(h,k-h)$ and $\mu^\sgctd_u(h,k-h)$ are well defined.
\end{corollary}

We experimented with the smooth generalized controlled model on the same 
setting as before but focusing on the phase transition region, that is, on
values of $A$ that are close to 14. The results reported in 
Figure~\ref{fig:smooth} were, thus, obtained varying $A$ from 10 to 18 (so 
$150 \leq m \leq 612$). It can be noted that the smooth model allows us to 
generate formulas that are precisely in the phase transition zone, moreover 
we can obtain harder formulas. Nonetheless, the smooth generalized controlled 
model remains less hard than the Chen-Interian model, if we compare the
hardest formulas that can be generated with the same number of existential 
variables, disregarding the number of universal variables. As we noted in
the main part of the paper, alternative ways to compare the hardness of 
the models may exist and finding them is an important open research question.

\renewcommand{\myOneCol}{0.93\columnwidth}
\begin{figure*}[t!] 
    \centering
    \includegraphics[width=\myOneCol]{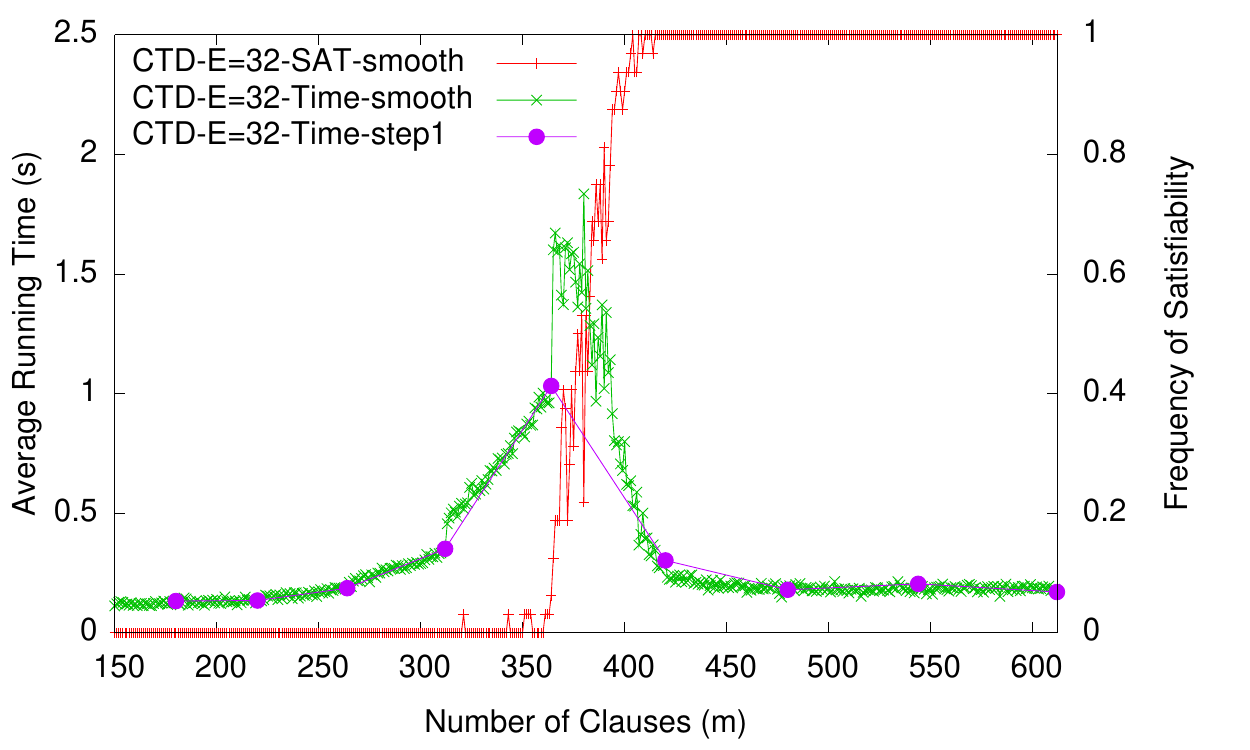}
    \caption{Comparing Generalized Controlled model with Smooth Generalized Controlled model.}\label{fig:smooth}
\end{figure*}

\subsection*{Additional notes on the generation of formulas}\label{app:generation}

Let $X$ be a set consisting of $N$ elements. We will consider the
following method to generate random elements of $X^t$ (the set of all 
$t$-tuples over $X$):
\begin{quote}
for each position $i$, $1\leq i\leq t$, select an element from $X$
uniformly at random.
\end{quote}
Clearly, every element of $X^t$ is equally likely to appear as the result
of this method. Thus, the method generates $t$-tuples over $X$ uniformly
at random. 

Let $S$ be a property of $t$-tuples over $X$ and let $p_N$ be the probability
that a $t$-tuple generated by the method described above has the property $S$. 
It follows that $p_N$ is the probability that an $t$-tuple selected from $X^t$ 
uniformly at random has the property $S$.

Next, let us define
\[
D^t(X) = \{\langle x_1,\ldots,x_t\rangle\in X^t \colon x_i\neq x_j,\
\mbox{for $i\neq j$}\}
\]
In other words, $D^t(X)$ is the set of all tuples in $X^t$ with 
no repeating elements.  

Let $S$ be a property of tuples in $X^t$. We will denote by $p'_N$ the 
probability that a tuple selected from $D^t(X)$ uniformly at random
has the property $S$. Then, if $t$ is sufficiently smaller than $N$, $p'_N$ 
can be closely estimated by $p_N$. To show that, let us define
\[
R^t(X) = X^t\setminus D^t(X).
\]
Clearly,
\[
p_N = \frac{|S|}{|X^t|}\quad\mbox{and}\quad p'_N=\frac{|D^t(X)\cap S|}{|D^t(X)|}.
\]
It follows that
\[
p_N - \frac{|R^t(X)\cap S|}{|X^t|}\leq p'_N \leq p_N + \frac{|R^t(X)\setminus S|}{|X^t|}
\]
and so,
\[
p_N - \frac{|R^t(X)|}{|X^t|} \leq p'_N \leq p_N + \frac{|R^t(X)|}{|X^t|}
\]
or, more explicitly,
\[
p_N - \big(1-\frac{(N-t+1)\cdots(N-1)N}{N^t}\big) \leq p'_N \leq p_N +\big(1- \frac{(N-t+1)\cdots(N-1)N}{N^t}\big).
\]
\begin{lemma}
If $\lim_{N\rar \infty} t^2/N = 0$, then 
\[
\lim_{N\rar\infty} \frac{(N-t+1)\cdots(N-1)N}{N^t} =1
\]
\end{lemma}

\begin{proof}: Clearly, 
\[
\Bigl(\frac{N-t}{N}\Bigr)^t \leq \frac{(N-t+1)\cdots(N-1)N}{N^t} \leq 1.
\]
Moeover, 
\[
\Bigl(\frac{N-t}{N}\Bigr)^t = \biggl[\Bigl(1 - \frac{1}{N/t}\Bigr)^{N/t}\biggr]^{t^2/N}
\]
Since $\lim_{N\rar \infty} t^2/N = 0$ and $t$ is a positive integer,
$\lim_{N\rar \infty} N/t = \infty$. Thus, 
\[
\lim_{N\rar\infty} \Bigl(1 - \frac{1}{N/t}\Bigr)^{N/t} = 1/e
\]
and, consequently,
\[
\lim_{N\rar\infty} \biggl[\Bigl(1 - \frac{1}{N/t}\Bigr)^{N/t}\biggr]^{t^2/N}
=1.
\]
\end{proof}

\begin{corollary}
If $\lim_{N\rar \infty} t^2/N = 0$, then there is a sequence $\varepsilon_N$ such that $\lim_{N\rar \infty}\varepsilon_N =0$
and
\[
p_N - \varepsilon_N \leq p'_N \leq p_N + \varepsilon_N.
\]
\end{corollary}

Next, we observe that if the property $S$ does not depend on the
order of the elements in a tuple in $D^t(X)$, that is, the probability 
that a tuple in $D^t(X)$ has the property $S$ is the same for every 
permutation of the elements in the tuple), then the probability that a 
set of $t$ elements from $X$ has a property $S$ (its ``set version'' 
to be precise) is given by $p'_N$.

Our earlier discussion shows then that to estimate the probability that 
a $t$-element subset of $X$ selected uniformly at random has a property 
$S$, it is sufficient to estimate the probability that a $t$-tuple over $X$ 
(an element of $X^t$) selected uniformly at random has the property $S$.
 
In this paper, we take advantage of this observation in the case when 
$X$ consists of formulas and $S$ is the property that a set (tuple) of
formulas is satisfiable ($\sat$), and unsatisfiable (UNSAT).

In particular, we consider in the paper the case when $X$ is the set of all 
non-tautological k-literal clauses over the set of $n$ 
propositional variables. We note that $|X|= 2^k{n \choose k}$.
It follows that when studying the probability that a $k$-CNF formula
with $m$ clauses is satisfiable, where $m=O(n)$, the results above apply and 
the probabillity, in the limit, is the same no matter whether we vew formulas
as sets or ordered tuples of clauses.

We also consider the case, when $X$ is the set of all $k$-CNF formulas with $m$
clauses over a set of $n$ variables, that is, the set $C(k,n,m)$. Also here, 
it makes no difference whether a disjunction of such formulas is considered a set 
of those formulas or an ordered tuple of such formulas. Since we consider disjunctions
of $t$ CNF formulas, where $t$ is fixed, the probability of such a disjunction being satisfiable
is, in the limit, not affected by how we interpret the disjunction --- as a set or an ordered tuple.

\end{document}